\documentclass[12pt]{article}
\usepackage[T1]{fontenc}
\usepackage{mathpazo,bbm}
\DeclareSymbolFont{calletters}{OMS}{cmsy}{m}{n}
\DeclareSymbolFontAlphabet{\mathcal}{calletters}
\ifdefined\useFiraSansFont
    \usepackage[T1]{fontenc}
    \usepackage[sfdefault,scaled=.85]{FiraSans}
    \usepackage{newtxsf}

\fi

\usepackage{microtype}
\usepackage{amssymb,amsmath,amsthm, bbm}
\usepackage{mathtools}
\usepackage{bm}        
\usepackage{thm-restate}
\usepackage[dvipsnames]{xcolor}
\usepackage{graphicx}
\usepackage{xspace}
\usepackage{multirow}
\usepackage{array}
\usepackage{complexity}
\usepackage[linesnumbered,ruled,vlined]{algorithm2e}
\SetKwInput{KwInput}{Input}                
\SetKwInput{KwOutput}{Output}
\SetKwFor{RepTimes}{repeat}{times}{end}
\SetKwComment{Comment}{/* }{ */}

\usepackage{dirtytalk}
\usepackage{tikz}
\usepackage[hypertexnames=false]{hyperref}
\usepackage[capitalize,nameinlink]{cleveref}

\usepackage[shortlabels]{enumitem}

\usepackage[margin=1in]{geometry}
\usepackage{braket}
\usepackage{ wasysym }
\hypersetup{colorlinks=true,urlcolor={MidnightBlue},linkcolor={MidnightBlue},
    citecolor={MidnightBlue}}
\usepackage{csquotes}

\theoremstyle{plain}
\newtheorem{theorem}{Theorem}[section]
\newtheorem{corollary}[theorem]{Corollary}
\newtheorem{proposition}[theorem]{Proposition}
\newtheorem{lemma}[theorem]{Lemma}
\newtheorem{claim}[theorem]{Claim}

\newtheorem{observation}[theorem]{Observation}

\newtheorem{definition}[theorem]{Definition}

\theoremstyle{remark}
\newtheorem{remark}[theorem]{Remark}
\theoremstyle{plain}
\newclass{\DNF}{DNF}
\newclass{\DNFs}{DNFs}
\newclass{\ACzero}{AC^0}
\newclass{\TCzero}{TC^0}
\newclass{\Logspace}{L}

\newcommand{\F}{\mathbb{F}} 
\newcommand{\Z}{\mathbb{Z}} 

\renewcommand{\Pr}{\mathop{\bf Pr\/}}
\renewcommand{\E}{\mathop{\bf E\/}}

\newcommand{\abs}[1]{\left|#1\right|}
\let\p\undefined
\newcommand{\p}[1]{\left(#1\right)}
\newcommand{\s}[1]{\left[#1\right]}
\newcommand{\cb}[1]{\left\{#1\right\}}

\newcommand{\ceil}[1]{\lceil #1 \rceil}

\newcommand{\floor}[1]{\lfloor #1 \rfloor}

\DeclarePairedDelimiter\normd{\lVert}{\rVert}
\newcommand{\norm}[1]{\normd*{#1}}
\newcommand{\bignorm}[1]{\Bigl \| #1 \Bigr \| }

\newfunc{\MAJ}{MAJ}
\newfunc{\pmaj}{promise\text{-}MAJ}
\newfunc{\MUX}{MUX}
\newfunc{\NAE}{NAE}
\newfunc{\OR}{OR} 
\newfunc{\AND}{AND}
\newfunc{\XOR}{XOR}
\newfunc{\Tribes}{Tribes}
\newfunc{\LocalCorrect}{LocalCorrect}

\newfunc{\sgn}{sgn} 

\newfunc{\spar}{sparsity}
\newfunc{\rank}{rank}
\newfunc{\spn}{span}
\newfunc{\quasipoly}{quasipoly}
\newfunc{\Bias}{Bias}
\newcommand{\noise}{\mathsf N}
\newcommand{\fouriernoise}{\mathsf T}

\newfunc{\DT}{DTdepth} 
\newfunc{\DTs}{DTsize} 

\renewcommand{\hat}{\widehat}
\renewcommand{\tilde}{\widetilde}

\newcommand{\bits}{\{0,1\}}

\newcommand{\calA}{\mathcal{A}}
\newcommand{\calB}{\mathcal{B}}
\newcommand{\calC}{\mathcal{C}}
\newcommand{\calD}{\mathcal{D}}

\newcommand{\calG}{\mathcal{G}}

\newcommand{\calO}{\mathcal{O}}

\newcommand{\calS}{\mathcal{S}}

\newcommand{\calU}{\mathcal{U}}

\newfunc{\Parity}{PARITY}

\newfunc{\Dict}{Dict}
\newfunc{\Corr}{Corr}
\newfunc{\avg}{avg}
\newfunc{\smooth}{smooth}
\newfunc{\dist}{\calD}

\newclass{\ETH}{ETH}

\renewcommand{\epsilon}{\varepsilon}

\newcommand{\indicator}{\mathbbm{1}}

\newcommand{\supp}{\mathrm{supp}}

\newfunc{\DISJ}{DISJ}
\newfunc{\search}{SEARCH}

\newfunc{\Th}{Th}
\newfunc{\coll}{Coll}

\newcommand{\pisucc}{\Pi_{\mathrm{succ}}}
\newcommand{\maxmix}{\pi_{\textrm{mix}}}

\newcommand{\range}{\mathsf{Y}}
\newcommand{\domain}{\mathsf X}
\newcommand{\dom}{\mathsf{dom}}

\newcommand{\oracle}{\calO}

\newcommand{\relation}{\mathsf R}
\newcommand{\inputspace}{D}

\newcommand{\outputspace}{\hat\range}
\newcommand{\inputreg}{\mathsf I}
\newcommand{\workreg}{\mathsf W}
\newcommand{\outputmapfn}{q}

\newcommand{\domainsize}{d}
\newcommand{\identity}{\,\mathbb{I}}

\newcommand{\op}[2]{\ket{#1}\!\bra{#2}}

\newcommand{\crefdefpart}[2]{%
    \hyperref[#2]{\namecref{#1}~\labelcref*{#1},~\ref*{#2}}%
}

\newcommand{\blank}[1]{}

\newcommand{\pispread}{\Pi_{\mathsf{spread}}}
\newcommand{\piout}{\Pi_{\tau}}

\newcommand{\toBin}[1]{\mathsf{Bin}_{#1}}
\newcommand{\toBins}[1]{\mathsf{Bins}_{#1}}
\newcommand{\counts}{\mathsf{counts}}

\newcommand{\sortproblem}{\textsc{Sort}}

\newcommand{\circuit}{\mathcal{C}}
\newcommand{\event}{\mathsf{E}}

\newcommand{\badset}{\ensuremath \event_{\mathsf{free}}}

\newcommand{\noisyx}{\tilde x}

\allowdisplaybreaks
\ifdefined\useTheoremBorders
\usepackage[most]{tcolorbox}

\tcbset{
    theorem border/.style={
        enhanced jigsaw,
        breakable,
        colback=white,
        colframe=black,
        boxrule=0.4pt,
        sharp corners,
        left=0.75em,
        right=0.75em,
        top=0.5ex,
        bottom=0.5ex,
        before skip=1.5ex plus 0.5ex minus 0.25ex,
        after skip=1.5ex plus 0.5ex minus 0.25ex,
        parbox=false,
    },
}

\newcommand{\borderTheoremEnvironment}[1]{%
    \tcolorboxenvironment{#1}{theorem border}%
}

\borderTheoremEnvironment{theorem}
\borderTheoremEnvironment{corollary}
\borderTheoremEnvironment{proposition}
\borderTheoremEnvironment{lemma}
\borderTheoremEnvironment{claim}
\borderTheoremEnvironment{fact}
\borderTheoremEnvironment{conjecture}
\borderTheoremEnvironment{exercise}
\borderTheoremEnvironment{question}
\borderTheoremEnvironment{obs}
\borderTheoremEnvironment{openquestion}
\borderTheoremEnvironment{claimNoNum}
\borderTheoremEnvironment{lemmaNoNum}
\borderTheoremEnvironment{theoremNoNum}
\borderTheoremEnvironment{corollaryNoNum}
\borderTheoremEnvironment{questionNoNum}
\borderTheoremEnvironment{definition}
\borderTheoremEnvironment{example}
\borderTheoremEnvironment{thm}
\borderTheoremEnvironment{open}
\borderTheoremEnvironment{remark}
\borderTheoremEnvironment{proof}

\fi

\newcommand{\out}[1]{\ket{#1}\!\bra{#1}}

\renewcommand{\supp}{\mathsf{Supp}}
\newcommand{\record}{\mathcal R}
\newcommand{\queryreg}{\mathsf Q}
\newcommand{\phasereg}{\mathsf P}
\newcommand{\functionreg}{\mathsf I}
\newcommand{\0}{\mathbf{0}}
\newcommand{\U}{\mathsf{U}}
\newcommand{\B}{\mathcal{B}}
\newcommand{\unif}[1]{\mathcal{U}_{\,#1}}
\newcommand{\recordingvs}[1]{\Gamma_{#1}}
\newcommand{\fouriervs}[1]{\mathsf{B}_{#1}}
\newcommand{\projrecordingvs}[1]{\Pi_{\recordingvs{#1}}}
\newcommand{\projfouriervs}[1]{\Pi_{\fouriervs{#1}}}
\newcommand{\eval}{\textsc{Eval}}
\newcommand{\hash}{\mathsf H}
\newcommand{\probspace}{\bm{\Delta}}

\NewDocumentCommand{\ezlem}{o +m}{%
    \IfNoValueTF{#1}
        {\begin{lemma}#2\end{lemma}}
        {\begin{lemma}[#1]#2\end{lemma}}%
}
\NewDocumentCommand{\ezprop}{o +m}{%
    \IfNoValueTF{#1}
        {\begin{proposition}#2\end{proposition}}
        {\begin{proposition}[#1]#2\end{proposition}}%
}
\NewDocumentCommand{\ezthm}{o +m}{%
    \IfNoValueTF{#1}
        {\begin{theorem}#2\end{theorem}}
        {\begin{theorem}[#1]#2\end{theorem}}%
}

\title{A Noise Operator Approach to Quantum Query Complexity and Time-Space Tradeoff Lower Bounds}
\author{Paul Beame\thanks{Research supported by NSF grant  CCF-2422205.}\\ Computer Science \& Engineering\\  University of Washington
    \and 
    Blake Holman\thanks{Research supported by the LDRD Program at Sandia National Laboratories. Sandia is managed and operated by
NTESS under DOE NNSA contract DE-NA0003525}\\
     Sandia National Laboratories\\
Purdue University
    \and
    Niels Kornerup\footnotemark[2]\\ Sandia National Laboratories}

\begin{document}

\maketitle

We introduce a novel method based on the noise operator for showing that time-bounded quantum query algorithms only have exponentially small success probability. 
We use this to prove strong quantum time-space tradeoff lower bounds that seem out of
reach by other methods.
A key feature of our noise operator method is that applying it involves purely classical arguments, which makes it particularly simple to use. 

We use our quantum noise stability bound to prove that for any strongly universal (pairwise independent) hash function family $\hash$ from $n$ bits to $m$ bits, almost all hash functions in $\hash$ require a quantum algorithm with at most $S$ qubits of memory to make $\Omega(nm/S)$ queries to an input $x$ in order to compute $h(x)$, even with very small success probability. Such a lower bound was not previously known even for classical algorithms. Previously, Mansour, Nisan, and Tiwari had used their hash mixing lemma to show the same quantitative bound for algorithms that have as input both a description of the hash function and its input. Our noise operator method allows us to use a related but simpler property of hash functions to prove our lower bounds.

As one consequence of this lower bound we derive the first optimal time-space tradeoff lower bounds for quantum algorithms computing integer multiplication.

We also use this method to prove the first fully general quantum time-space tradeoff lower bound for sorting.
Previous arguments for sorting, including the first quantum time-space tradeoff lower bounds proven by Klauck, \v{S}palek and de Wolf, are limited to proving output-oblivious lower bounds (i.e. the lower bounds only apply to algorithms with a non-adaptive output schedule) -- subsequent methods, despite being useful for deriving fully general lower bounds for other problems, have yielded nothing beyond output-oblivious lower bounds for sorting. 

More precisely, by combining our quantum noise stability bound with quantum recording query methods, we prove an $\Omega(n^{4/3} (\log \log n)/(S^{1/3} \log n))$ lower bound on the number of queries that a fully general quantum algorithm with at most $S$ qubits of memory requires to sort $n$ numbers from $[n^2]$.

\section{Introduction}

Running time and the amount of memory (space) that algorithms use are among the most important
resources used in classical computation.
The challenges we face in quantum computation make both of these resources paramount:  
We need quantum computations to be fast, since keeping them isolated from outside entanglements gets harder
as time progresses, and we need the number of qubits that we maintain in simultaneous superpositions to be small,  because each logical qubit corresponds to such a large number of physical qubits and because maintaining large numbers of them is expensive.
Because of these costs, it is even more critical for quantum computation to optimize both of these resources and understand what the possible tradeoffs are between them. 

Methods for proving time-space tradeoff lower bounds for quantum computation have been known for roughly twenty years beginning with the work of Klauck, \v{S}palek and de Wolf~\cite{KSdW07}, who showed how to
translate the major paradigm for proving classical lower bounds due to Borodin and Cook~\cite{BC82} to quantum computation.   The Borodin-Cook method is based on proving time lower bounds for query algorithms that
have only an exponentially small probability of 
success. 
In classical computation, the method is able to focus on the success probability on individual computation paths, since these paths cannot interfere with one another.
This is obviously not possible in analyzing quantum query complexity. 
Klauck, \v{S}palek, and de Wolf showed that the same connection to query complexity with exponentially small success probability holds for quantum computation, but they needed a different analysis technique to prove exponentially small success probability of quantum query algorithms.

The technique that Klauck, \v{S}palek and
de Wolf used to obtain lower bounds for quantum
query algorithms with exponentially small success probability was \emph{strong direct product theorems}.
They proved such theorems for functions such as OR to obtain time-space tradeoff lower bounds by showing how to embed direct-product problems
into the outputs of sorting and
Boolean matrix product to obtain the first time-space tradeoff lower bounds for quantum computation.
In particular, they proved the well-known lower bound of $T=\Omega(n^{3/2}/\sqrt{S})$ for sorting which 
is tight up to a polylog factor in the quantum query model~\cite{DBLP:conf/stoc/Klauck03}.

However, this method has a drawback that is not so well known:   In order to prove a lower bound, one needs to know the output positions in which to embed the direct-product problems, independent of the input. 
It follows that lower bounds proven by this method only apply to \emph{output-oblivious} quantum algorithms that are forced to commit to a schedule for producing their outputs that is oblivious to the input.
In particular, the Klauck, {\v{S}}palek, and  de Wolf~\cite{KSdW07} sorting lower bound requires that the outputs always be produced in ascending rank order with fixed output time-steps.
This fixed order requirement for output production, though not the specific time steps, is similar to the optimal classical lower bound of Beame \cite{DBLP:journals/siamcomp/Beame91}, but is more conditional than a full general classical time-space
tradeoff of $T=\Omega(n^2/(S \log n))$ of \cite{BC82} which does not have this requirement.

A fully general quantum time-space tradeoff lower bound for sorting should allow algorithms to adaptively choose the order in which they wish to produce outputs based on the results of their queries.
In contrast, an output-oblivious sorting algorithm might have to, for example, commit to the value for the smallest element even if it has only observed values it expects to be high ranking.
Hamoudi and Magniez~\cite{HM23} re-proved the same $T=\Omega(n^{3/2}/\sqrt{S})$  ascending-rank output-oblivious lower bound for sorting, and Beame and Kornerup~\cite{bk:cumulative-journal} showed how to extend this lower bound to arbitrary output-oblivious orderings.
Nonetheless, the problem of proving any non-trivial fully general quantum time-space tradeoff lower bound for sorting has remained wide open!

To resolve this it would be natural to try to extend the classical fully general lower bound of~\cite{BC82}, but one first needs a lower bound method that is not inherently limited
to output-oblivious algorithms.
Hamoudi and Magniez~\cite{HM23} were able to prove the first time-space tradeoff lower bounds for general quantum algorithms by introducing another method, based on the \emph{recording-query basis} (a.k.a. compressed oracle) approach of Zhandry~\cite{zhandry}.
This method lets one express the results of
quantum query algorithms running for $t$ steps as
superpositions over basis states that are defined by
a set of at most $t$ input coordinates where the values
are exactly known and the remaining input coordinates are the same as in the original input distribution (denoted by $\perp$ in the basis).   While these basis states have something in common with
classical computation paths of length $t$, they occur in superposition and must be considered together.
The general approach in~\cite{HM23} is to identify a measure of
progress for recording-query basis states based on their
known coordinates, prove inductively that the total amplitude on basis states with substantial progress is exponentially unlikely, and then to prove that producing correct answers on combinations of basis states without
substantial progress is exponentially unlikely.   
They used this to prove a general quantum time-space tradeoff lower bound
for the problem of finding multiple collisions in an input.  
Unfortunately, as noted above, when applying this method to sorting, Hamoudi and Magniez were only able to prove lower bounds for output-oblivious algorithms.  

Beame, Kornerup, and Whitmeyer~\cite{BKW26} introduced an 
alternative method to obtain fully general time-space tradeoff lower bounds for a wide variety of linear algebra problems, many of them optimal, using the recording-query basis approach,
by \emph{bucketing} recording-query basis states based
on shared coordinates with value $\perp$, indicating that no information about those coordinates has been learned.   
Strong lower bounds for individual buckets follow from similar ideas to classical lower bounds that do not know anything about those shared input coordinates, but the states in different buckets are not mutually orthogonal so the method requires that the total number of different buckets is small.
This number will only be small if the defining sets of
shared $\perp$-coordinates are small; this is sufficient
for the lower bounds for linear-algebra problems.
However, the fully general classical lower bound for sorting in \cite{BC82} requires that a constant fraction of the input
coordinates have not been queried in order to prove
exponentially small success probability, which rules
out the bucketing approach also.

In this paper we find a way to make a quantum version of the classical time-space tradeoff lower bound for sorting in \cite{BC82}.
The result is the first non-trivial fully general quantum time-space tradeoff lower bound for sorting.
To do this, we combine the recording-query progress
approach in \cite{HM23} with a novel approach that takes advantage of the \emph{noise operator} from the analysis of Boolean functions (see \cite{ODonnell14} for a detailed overview of the topic).
The result is a fully general $T=\Omega(n^{4/3} (\log \log n)/(S^{1/3} \log n))$ quantum query lower bound for sorting.

Once established, this \emph{quantum noise stability bound} argument is surprisingly general and easy to apply; it lets us upper bound the success probability of any quantum query algorithm based on the probability that applying a small amount of noise to \emph{typical} inputs will cause the output to change.
Applying the method only requires purely classical quantities.
While the method does not result in tight lower bounds for all functions, it is quite natural to apply and is capable of producing non-trivial lower bounds even for small success probabilities.
This makes our method a natural first trick to apply when trying to prove new quantum time-space tradeoff lower bounds.

We use our quantum noise stability bound alone to prove a fully general quantum time-space tradeoff lower bound that improves on one of the most general lower bounds known for classical computation and is the first non-trivial quantum time-space tradeoff lower bound for the problem: 
Mansour, Nisan, and Tiwari~\cite{DBLP:journals/tcs/MansourNT93} proved that for every strongly universal (i.e., pairwise independent) hash function family $\hash$ from $n$ bits to $m$ bits, the problem of computing $h(x)$ on input $(h,x)$ for $h\in \hash$
and $x\in \bits^n$ on classical computers with space bound $S$
requires time $T=\Omega(mn/S)$.
The bound is based on their Hash Mixing Lemma, which essentially shows that as long as $\Omega(n)$ bits of the hash function and its input are unspecified, the output values are not too likely to lie in any fixed
subset.

Proving any non-trivial quantum lower bound for this problem seems out of reach using previous methods:
Neither the recording query progress method in \cite{HM23} nor the bucketing method of \cite{BKW26} seem to be able to take advantage of the Hash Mixing Lemma, as there is no natural notion of \emph{unusual progress} for this problem and, as was the case for sorting, the requirement of $\Omega(n)$ unspecified input bits means that the number of shared $\perp$ coordinates needed to define the buckets would be too large to guarantee that the number of buckets is small.

Using our quantum noise stability method, 
We prove that 
\emph{almost all} hash functions $h\in \hash$ \emph{individually} require
time $T=\Omega(nm/S)$  lower bound for fully general quantum (or classical) algorithms
to compute (given only input $x\in \bits^n$) for any space bound $S$.
This not only extends the quantitative lower bound
to fully general quantum algorithms, it is qualitatively much stronger than the result of~\cite{DBLP:journals/tcs/MansourNT93}, which only proves the collective hardness of
computing a strongly universal hash function family $\hash$ -- only holding when the hash function $h$ is part of the input.

To derive our bound we prove an alternative kind of hash mixing lemma which implies that for every choice of $k$, almost all functions in a strongly universal hash function family have the property that almost all inputs to that hash function are extremely sensitive to noise on all possible choices of $k$ output indices.
This purely classical property is sufficient to derive the $T=\Omega(mn/S)$ quantum time-space tradeoff lower bound for universal hashing with our noise method.
As an immediate consequence of this bound, we improve the quantum time-space tradeoff lower bound of $\Omega(n^2/(S \log n))$ for integer multiplication in \cite{BKW26} by a factor of $\log n$.

We also show that our quantum noise stability bound enables us to provide a significantly streamlined proof  showing that, for any fixed rigid $n\times n$ matrix $A$,
fully general quantum algorithms with space bound $S$ require time $T=\Omega(n^2/S)$ to
evaluate $Ax$ for $x\in \inputspace^n\subseteq \mathbb{F}^n$,
which matches the optimal lower bound of Beame, Kornerup, and Whitmeyer~\cite{BKW26} for the problem when $|D|$ is constant.

\begin{sloppypar}
\paragraph{Road-map}
In \cref{sec:noise} we introduce and prove our \emph{quantum noise stability bound} (\cref{cor:noise-stability}) and explain how it can be used in proving quantum time-space tradeoff lower bounds.
In \cref{sec:hashing} we prove our $T=\Omega(nm / S)$ quantum time-space tradeoff lower bounds for computing almost all members of strongly universal hash function families.
\cref{sec:sorting} contains our $T=\Omega(n^{4/3} (\log \log n)/(S^{1/3} \log n))$ fully general quantum time-space tradeoff lower bound for sorting, 
the proof of which follows by combining recording query methods to derive
properties that we show
in \cref{sec:spreadrecord,sec:nonspread}
with properties of the noise operator that we
show in 
\cref{sec:noisysort}.
Finally in \cref{sec:other-bounds} we give some intuition about when our quantum noise stability bound is useful for proving quantum time-space tradeoff lower bounds and show how it can be applied to 
prove
a version of 
the matrix-vector product lower bound 
originally shown in \cite{BKW26}.
\end{sloppypar}

\section{Preliminaries}
We start with some basic facts that we will need in this paper.

\begin{proposition}[Perron-Frobenius]\label{thm:perron-frobenius-row-sums}
      Let $M\in\mathbb R^{N\times N}$ be symmetric and entrywise
        nonnegative, and let
        $r_i=\sum_{j=1}^N M_{ij}$.
        Then
       $\min_{i\in[N]}r_i
            \le
            \lambda_{\max}(M)
            \le
            \max_{i\in[N]}r_i$.
        In particular, if every row sum equals $r$, then
        $\lambda_{\max}(M)=r$.
\end{proposition}

\begin{proposition}[Chernoff lower bound]\label{thm:chernoff-lower-tail}
    Let $R\sim\operatorname{Binomial}(n,q)$ and $\mu=nq$. For every
        $0\le\delta\le 1$,
        $
            \Pr\s{R\le(1-\delta)\mu}
            \le
            e^{-\delta^2\mu/2}.
        $
        In particular,
        $
            \Pr\s{R\le \mu/2}
            \le
            e^{-\mu/8}$.
\end{proposition}

\subsection{Unitary quantum query circuits}\label{sec:unitary-circuit}
We are interested in quantum circuits that compute functions $f:\domain^n \to \range^m$ or relations $\tau \subseteq  \domain^n \times \range^{\subseteq [m]}$ corresponding to computing some number of coordinates of $f$.
We let $\domainsize = \abs{\domain}$ and when $\domain \neq [d]$, we assume the existence of a bijective map $v$ between $\domain$ and $[\domainsize]$ which gives a total ordering on $\domain$.
To simplify presentation, we will use $x$ in lieu of $v(x)$ whenever the intended meaning is clear from context.
Without loss of generality, we will think of the state of a quantum algorithm as being comprised of a query register $\queryreg$, a phase register $\phasereg$, a workspace register $\workreg$, and an input register $\inputreg$.
We define input oracle $\oracle$ that maps each basis state $\oracle \ket{i,p,w}_{\queryreg,\phasereg,\workreg}\ket{x}_\inputreg = \omega_\domainsize^{p x_i} \ket{i,p,w}\ket{x}$ where $i,p,w,x$ are the values stored in $\queryreg,\phasereg,\workreg,\inputreg$ respectively and $\omega_\domainsize$ is some fixed $\domainsize$'th root of unity.
We will drop the register subscripts when they are obvious from context.

A $T$ query quantum circuit $\circuit$ is then defined using $T+1$ unitaries $U_0, \ldots, U_T$ acting only on $\queryreg, \phasereg, \workreg$.
The state of $\circuit$ after $t \leq T$ queries on input $x$ is given by the pure state
\begin{displaymath}
    \ket{\psi_t} = \ket{\psi^x_t}_{\queryreg,\phasereg,\workreg}\ket{x}_\inputreg = (U_t \otimes \identity) \oracle (U_{t-1} \otimes \identity) \oracle \ldots (U_1 \otimes \identity) \oracle (U_0 \otimes \identity) \ket{0}_{\queryreg \phasereg \workreg}\ket{x}_{\inputreg}.
\end{displaymath}
For any distribution $\dist$ over domain $\inputspace$, let $\ket{\dist} = \sum_{x \in \inputspace} \sqrt{\dist(x)}\ket{x}$.
Since each $U_t$ and $\oracle$ behave as the identity on $\inputreg$, we can simulate running $\circuit$ on input distribution $\dist$ by initializing $\inputreg$ to $\ket{\dist}$.
Thus after $t$ queries the state of our algorithm on input distribution $\dist$ is given by the pure state
\begin{displaymath}
    \ket{\psi_t} = \sum_{x \in \inputspace} \sqrt{\dist(x)}\ket{\psi_t^x}\ket{x} = (U_t \otimes \identity) \oracle (U_{t-1} \otimes \identity) \oracle \ldots (U_1 \otimes \identity) \oracle (U_0 \otimes \identity) \ket{0}_{\queryreg \phasereg \workreg}\ket{\dist}_{\inputreg}.
\end{displaymath}
After final query $T$, $\inputreg$ is measured in the standard basis to sample the input $x \sim \dist$ and the final state of the algorithm collapses to state $\ket{\psi^x_T}$.
The work register $\workreg$ of $\ket{\psi^x_T}$ is then measured in the standard basis, producing classical bit-string $w$.
The output of $\circuit$ is then given by some input independent post processing function $\outputmapfn(w)$.
There are two notions of output for circuits that we will be using in this paper.

\paragraph{Computing (restricted outputs of) functions}
We start by formally defining restricted outputs of functions.
\begin{definition}
    Let $f: \domain^n \to \range^m$ be a function.
    For any $S \subseteq [m]$ we use $f_S$ to denote $f$ restricted to its output coordinates with indices in $S$.
\end{definition}
Since the set $S$ is fixed, computing the restricted outputs $S$ of a function $g: \domain^n \to \range^m$ is the same as computing the function $f: \domain^n \to \range^S$ given by $f(x) = g_S(x)$. 

When $\circuit$ is tasked to produce the full output (or restricted outputs) of a function, we say that it has correctly computed input $x$ if $q(w) = f(x)$.
Then we can define $\Pi_{\outputmapfn(w)} = \sum_{x \in f^{-1}(\outputmapfn(w))} \op{x}{x}$ and $\Pi_{\text{succ}} = \sum_{w} \identity_\queryreg \otimes \identity_\phasereg \otimes \left(\op{w}{w}\right)_{\workreg} \otimes (\Pi_{\outputmapfn(w)})_{\inputreg} $ and the success probability of $\circuit$ on input distribution $\dist$ is given by $\norm{\Pi_{\text{succ}}\ket{\psi_T}}^2$.

\paragraph{Computing partial assignments}
We start by formally defining partial assignments.
\begin{definition}
    Let $f: \domain^n \to \range^m$ be a function.
    We say that $\tau \in \range^{\subseteq [m]}$ is a partial assignment that agrees with $f(x)$ (denoted $f(x) || \tau$) if $\tau_i = f_i(x)$ on every point where $\tau$ is defined.
    The size of $\tau$ is the number of points on which it is defined.
\end{definition}
Informally, one can think of partial assignments as a generalization of restricted outputs where the choice of set $S$ is not fixed.

When $\circuit$ is tasked to produce a partial assignment to $k$ outputs of $f$, we say that it has correctly produced such a partial assignment if $q(w)$ is some vector $\tau \in \range^S$ for a set $S$ of size $k$ such that $f(x) || \tau$.
We can define $\Pi_{q(w)} = \Pi_{\tau} = \sum_{x \text{ s.t. } f(x) || \tau} \op{x}{x}$ and $\Pi_{\text{succ}} = \sum_{w} \identity_\queryreg \otimes \identity_\phasereg \otimes \left(\op{w}{w}\right)_{\workreg} \otimes (\Pi_{\outputmapfn(w)})_{\inputreg}$.
Then the success probability of $\circuit$ on input distribution $\dist$ is given by $\norm{\Pi_{\text{succ}}\ket{\psi_T}}^2$.

\subsection{Space bounded quantum circuits}
Typically the deferred measurement principle~(\cite{DBLP:books/daglib/0046438}) lets one assume without loss of generality that all measurements are postponed to after the final step of a quantum algorithm.
However, doing this comes at a cost in terms of the space cost of the circuit, as each `delayed measurement' must be stored in a fresh ancilla qubit.
In particular, as is common in other works (both classical and quantum) on time-space tradeoffs (e.g. \cite{BC82,Abr91,DBLP:journals/tcs/MansourNT93,KSdW07,HM23,bk:cumulative-journal,BKW26}), it is useful to let algorithms write parts of the output to a write-only tape which does not count against their space bound in the middle of the computation.

To account for both of these factors, we allow space bounded quantum algorithms to perform intermediate measurements.
We take note of two subtly different models for these intermediate measurements.
An \emph{output oblivious} quantum circuit, as first used in \cite{KSdW07}, can be formalized as having $m$ fixed time-steps $t$ where fixed output index $i_t$ must be produced.
After applying $U_t$ a measurement is applied on the first $\ceil{\log_2 \abs{\range}}$ qubits of $\workreg$ in the standard basis, obtaining value $w$ and committing to $\outputmapfn(w, i_t)$ as the $i_t$'th output.
A \emph{fully general} quantum circuit, as first specified in \cite{HM23}, allows the algorithm to be adaptive with when it chooses to produce each output.
This is done by adding a `flag' measurement gate in the standard basis on the first qubit of $\workreg$ after applying each $U_t$.
If that measurement gives value $1$, then a second conditional measurement gate will measure the next $\ceil{\log_2 m} + \ceil{\log_2\abs{\range}}$ qubits of $\workreg$ in the standard basis, obtaining value $w$ and committing to $\outputmapfn(w)$ as a partial assignment to $1$ output of $f$.
For full generality, we allow both output oblivious and non output oblivious circuits to include additional measurement gates as a means to reset qubits to $\ket{0}$ without requiring uncomputation.
Since we make no distinction between quantum and quantum + classical space in this paper, we assume that the state of $\circuit$ is entirely stored in qubits.

It is worth noting that this model is much more messy than our unitary model in \cref{sec:unitary-circuit}.
Intermediate measurements make our algorithm's state mixed and might interfere with many of the lower bound tools we use in this paper.
Fortunately, we are able to have the best of both worlds.
We prove our time-space tradeoff lower bounds via key lemmas about query algorithms producing partial outputs that do not require bounded space.
Since the deferred measurement principle gives us that intermediate measurements could not possibly help such algorithms, we can prove these lemmas in the unitary model and then directly apply them to our bounded space models.

For the output oblivious model, these lemmas prove exponentially small upper bounds (in $k$) on the probability of computing $f_S$ for fixed choices of $S \subseteq [m]$ with $\abs{S} = k$ after $h(n,k)$ quantum queries to the input.
For the non output oblivious model, these lemmas instead prove exponentially small upper bounds (in $k$) on the probability of producing a partial assignment to $k$ outputs of $f$ after $h(n,k)$ quantum queries to the input.

A key ingredient in converting these lemmas into time-space tradeoff lower bounds is the following result implicitly proven in \cite{Aar05}, first applied to this context in \cite{KSdW07}, and fully formalized and proven as a stand-alone statement in \cite{BKW26}.
\begin{proposition}[\cite{Aar05,KSdW07,BKW26}]\label{thm:quantum-union-bound}
    Let $\calC$ be a quantum circuit, $\rho$ be an $S$-qubit (possibly mixed) quantum state, and $\maxmix$ be the maximally mixed state on $S$ qubits.
    If $\calC$ produces output $z$ on input state $\rho$ with probability $p$, then $\calC$ produces output $z$ on input state $\maxmix$ with probability at least $p 2^{-S}$.
\end{proposition}
This proposition lets us apply our key lemmas to blocks of a quantum algorithm that are initialized with $S$ qubits of input-dependent advice resulting from prior queries.

\subsection{Recording Queries}

    A query algorithm has a query register $\queryreg$, a phase register
    $\phasereg$, and a workspace register $\workreg$, with computational-basis
    states
    \[
        \ket{i,p,w}_{\queryreg\phasereg\workreg},
        \qquad
        i\in[n],\quad p\in\Z_m.
    \]
    The recording database register $\inputreg$ has basis states
    \[
        \ket {z}_\inputreg
        =
        \bigotimes_{i=1}^n\ket{z_i}
        \qquad
        z\in(\Z_m\cup\{\bot\})^n.
    \]
We describe the behavior of the recording query operator for a single input register first. 
For $w\in \Z_m$ recall that 
    \[
        \ket{\hat w}
        =
        \frac1{\sqrt m}\sum_{x\in [m]}\omega_m^{w\,x}\ket x.
    \]
    Let $\mathsf S$ swap $\ket\bot$ and $\ket{\hat0}$ and fix their
    orthogonal complement. 
    For each phase $p\in\mathbb Z_m$, define
    \[
        \oracle_p\ket z=\omega_m^{p\, z}\ket z,
        \qquad
        \oracle_p\ket\bot=\ket\bot,
        \qquad
        \record_p=\mathsf S\oracle_p\mathsf S.
    \]
    Thus $\record_0=\identity$, while for $p\ne0$,
    \[
        \record_p\ket\bot=\ket{\hat p},
        \qquad
        \record_p\ket{\hat w}
        =
        \begin{cases}
            \ket{\hat0},         & w=0,             \\
            \ket\bot,            & w=-p,            \\
            \ket{\widehat{w+p}}, & w\notin\{0,-p\},
        \end{cases}
    \]
    with arithmetic in $\Z_m$.

The extension to all $n$ input registers is straightforward as follows:
The full oracle $\oracle$ in the recording  basis is given by
$\oracle=\sum_{i\in[n]}\,
        \sum_{p\in\Z_m}
        (\ket{i,p}\!\bra{i,p})
        \otimes \identity_{\workreg}
        \otimes \oracle_{i,p}$
where $\oracle_{i,p}=\identity^{\otimes (i-1)}\otimes \mathcal{O}_p\otimes \identity^{\otimes (n-i)}$.
Write  $\record_{i,p}=\mathsf S^{\otimes n} \mathcal{O}_{i,p} \mathsf S^{\otimes n}$.
Then, letting $\mathcal{S}=\identity_{\queryreg\phasereg\workreg}\otimes \mathsf{S}^{\otimes n}$, we have
the full recording query operator
$\record=\mathcal{S}\mathcal{O}\mathcal{S}$
and hence
    \[
        \record
        =
        \sum_{i\in[n]}\,
        \sum_{p\in\Z_m}
        (\ket{i,p}\!\bra{i,p})
        \otimes \identity_{\workreg}
        \otimes \record_{i,p}.
    \]

\begin{proposition}[Lemma 4.1 in \cite{HM23}]
\label{prop:recording-query}
If the recording-query operator $\mathcal{R}$  is applied to a recording-query
basis state $\ket{i,p,w}\ket{z_1,\ldots,z_n}$  where $p \ne 0$ then the register $\ket{z_i}$ is mapped to
\begin{displaymath}
    \begin{cases}
        \sum\limits_{y\in [m]}\frac{\omega_m^{p\, y}}{\sqrt{m}}\ket{y}&\textrm{if }z_i=\bot\\
        (1-\frac{2}{m})\,\omega_m^{p\, z_i}\ket{z_i} + \frac{1}{m} \ket{z_i} +  \frac{\omega_m^{p\, z_i}}{\sqrt{m}} \ket{\bot} + \hspace{-5pt} \sum\limits_{y\in [m]\setminus \{z_i\}} \hspace{-5pt} \frac{1-\omega_m^{p\, y}-\omega_m^{p\, z_i}}{m}\ket{y}&\textrm{otherwise.}
    \end{cases}
\end{displaymath}
If $p = 0$ then the register remains unchanged.
\end{proposition}

    A recording-query algorithm has initial
    state $\ket{0,0,0}_{\queryreg\phasereg\workreg}\ket{\bot}^{\otimes n}_\inputreg$ and alternates unitaries $\U=\U_{\queryreg\phasereg\workreg}\otimes \identity_\inputreg$ with recording queries.
Consequently, the algorithmic unitaries
    commute with every recording-basis projector, and after $t$ queries the
    database contains at most $t$ non-$\bot$ entries.

    \begin{proposition}[\cite{zhandry}]
    \label{prop:tquery}
      The state of any quantum query algorithm that makes
      at most $t$ queries to its input is a linear
      combination of states whose input register in the recording-query basis $\ket{z}$ has the property
     that $z$ contains at most $t$
      non-$\bot$ entries.
    \end{proposition}

        \begin{definition}
        Let $\recordingvs{t}$ denote the subspace over the input registers spanned by recording-query basis vectors $\ket{z_1, \ldots, z_n}$ for $z_i \in[m]\cup\{\bot\}$ such that at most $t$ of the $z_i \neq \bot$.
        Define $\projrecordingvs{t}$ to be the projection
        onto $\recordingvs{t}$.
    \end{definition}

\paragraph{Alternative formulation in the Fourier basis}

We focus on the domain $\domain=\Z_m$ and the input distribution being the uniform
distribution on $\domain^n$.

A quantum query algorithm (with phase queries) has a query register $\queryreg$, a phase register
    $\phasereg$, and a workspace register $\workreg$.
    When we purify this by adding the input register
    $\inputreg$, we obtain computational basis
    states
    \[
        \ket{i,p,w}_{\queryreg\phasereg\workreg}\ket{x}_\inputreg,
        \qquad
        \textrm{for 
        $i\in[n]$, $p\in\Z_m$, $x\in [m]^n$.}
    \]

If $n=1$, for each single $z\in \Z_m$ we have the associated Fourier basis state over $\inputreg$ 
    \[
        \ket{\hat z}
        =
        \frac1{\sqrt m}\sum_{x\in [m]}\omega_m^{z\,x}\ket x.
    \]
More generally, when $n\ge 1$ and $z\in \Z_m^n$ we
have the associated Fourier basis state over $\inputreg$ as
$$ \ket{\hat z}
        =\bigotimes_{i\in [n]} (\frac1{\sqrt m}\sum_{x_i\in [m]}\omega_m^{z_i\,x_i}\ket {x_i})=\frac{1}{m^{n/2}}\sum_{x\in [m]^n} \omega_m^{z\cdot x} \ket{x}.$$
        
Observe that for a single coordinate of the input register, $\ket{\hat 0}$ is the analogue of $\perp$ in the 
recording query basis, and that for the vector $\0\in \Z_m^n$ consisting of $n$ zeroes,
the initial state of any quantum query algorithm
over $\domain^n=[m]^n$ under the uniform distribution on inputs is
$$\ket{0,0,0}_{\queryreg\phasereg\workreg}\ket{\hat \0}_\inputreg.$$

By applying this translation, we can rewrite~\cref{prop:tquery} equivalently as

\begin{corollary}
    \label{cor:tquery-fourier}
       The state of any quantum query algorithm that makes
      at most $t$ queries to its input is a linear
      combination of states whose input register in the Fourier basis $\ket{\hat z}$ has the property
     that $z$ contains at most $t$
      non-$0$ entries.
\end{corollary}

\begin{definition}
Let $\fouriervs{t}$ denote the subspace of the input registers spanned by Fourier basis vectors $\ket{\hat{z}}=\ket{\hat{z_1}, \ldots, \hat{z_n}}$ for $z_i \in [m]$ such that at most $t$ of the $z_i \neq 0$.
Elements of $\fouriervs{t}$ are said to have \emph{Fourier degree at most} $t$.
Define $\projfouriervs{t}$ to be the projection
onto $\fouriervs{t}$.
\end{definition}

\begin{proposition}
For every $t$,
    $\projfouriervs{t}=\mathsf S^{\otimes n}\, \projrecordingvs{t}\,\mathsf S^{\otimes n}$ and
    $\projrecordingvs{t}=\mathsf S^{\otimes n}\, \projfouriervs{t}\,\mathsf S^{\otimes n} $.
\end{proposition}

\section{Noise operators and quantum query algorithms}\label{sec:noise}

    The compressed oracle model using the recording query basis is a powerful tool for analyzing the knowledge of a quantum query algorithm by keeping track of a database of queries and answers intuitively corresponding to the input values potentially learned by the algorithm. The problem, however, is that because it can make these queries in superposition, a quantum algorithm can potentially learn much more about the input than just the values
    of the coordinates it has queried because, unlike classical randomized algorithms, the different databases can be entangled with each other. 
    While there are valuable tools for upper bounding the success probability in the recording query basis \cite{zhandry, parallelqrom}, they fail to produce inverse-exponential upper bounds. 
    It is for this reason that \cite{HM23} developed problem-specific analyses for bounding the algorithm's probability of success in the standard basis (where the input being queried actually lives) given that the state in the recording query basis has some property. 
    Still, these techniques often fail when the witness size for the property is too large.
    
    We show how to prove inverse-exponential bounds under the uniform input distribution over inputs of the form $[m]^n$ (which is also the typical setting for the recording query basis) for problems that are somewhat sensitive to the input, using an approach that ends up requiring predominantly
    classical reasoning.

    By~\cref{cor:tquery-fourier}, after $t$ queries, the input register of the state $\ket{\psi_t}$ has Fourier degree at most $t$.
    However, to find the probability that some event $E$ has occurred, we care about properties of $\ket{\psi_t}$ in the standard basis.
    We therefore aim to prove an upper bound on $\norm{\Pi_E \projfouriervs{t}}^2$. Our first observation is that this quantity is equal to $\norm{\Pi_E\projfouriervs{t}\Pi_E}$, which quantifies how stable the event $E$ can be when supported by low-degree states. 
    Unfortunately, the analysis of our events of interest in the standard basis can be very difficult, since the natural basis of the subspace $\fouriervs{t}$ is the set of low-Fourier degree vectors. 
    Instead, we ``smooth out'' $\projfouriervs{t}$ by applying noise. 
    This ``noisy'' version of $\projfouriervs{t}$ can be interpreted as a random walk over inputs, where each entry is resampled with some probability.

    We review the ideas related to noise sensitivity 
    and the noise operator from~\cite{ODonnell14}, though we focus on input coordinates having values in $[m]$
    rather than Boolean values:
   For each
    coordinate $x_i\in [m]$, 
    the random variable $\noise_p(x_i)$ on $[m]$ is given by a two-stage process:  with probability $p$, replace $x_i$ with a random $y\sim\unif{[m]}$ and
    with probability $1-p$ return $x_i$.
    Observe that applying $\noise_p$ does not change the distribution over inputs that were originally uniformly random.
    
    This can be related to the following single-coordinate noise operator over
    the Fourier basis
        \[
            \fouriernoise_p
            =
            p\ket{\hat0}\!\bra{\hat0}+(1-p)\identity,
        \]
 which can be interpreted as a Markov transition matrix
 from $x$ to $\noise_p(x)$, as follows:
 \begin{proposition}
 For any $x, y\in [m]$,      
 $$\Pr[\noise_p(x)=y]=\bra x \fouriernoise_p \ket y=(1-p)\,\delta_{x,y}+\frac{p}{m}$$
 where $\delta_{x,y}$ is the Kronecker $\delta$ function that equals 1 iff $x=y$. 
  \end{proposition}
        
For inputs $x\in [m]^n$ the random variable $\noise_p(x)$ is defined by applying $\noise_p(x_i)$ at every
index independently: $\noise_p(x)=\bigotimes_{i\in [n]}\noise_p(x_i)$.
The corresponding noise operator over the Fourier basis is therefore the $n$-fold tensor product of the single copy case $\fouriernoise_p^{\otimes n}$.
Therefore, the probability that $\fouriernoise^{\otimes n}_p$ transitions $x\in [m]^n$ to some $y\in [m]^n$ is \[
                \bra y \fouriernoise^{\otimes n}_p\ket x
                =
                \Pr\s{\noise_p(x)=y}
                =
                \p{1-p+\frac{p}{m}}^{n-d(x,y)}
                \p{\frac{p}{m}}^{d(x,y)}\]
                where $d(x,y)=\sum_{i\in [n]} (1-\delta_{x_i,y_i})$ is the Hamming distance
                between $x$ and $y$.

        Understanding the relationship between $\fouriernoise^{\otimes n}_p$ and $\projfouriervs{t}$ is straightforward since they are both diagonalizable in the Fourier basis. First, we have already defined: 
                \[
            \projfouriervs{t}
            =
            \sum_{\abs{\supp'(z)}\le t}
            \ket{\hat z}\!\bra{\hat z}.
        \]
        We also have 
        \begin{align*}
            \fouriernoise^{\otimes n}_p
            &= \bigotimes_{i\in [n]}\bigg(\out{\hat 0} + (1-p)\sum_{z_i\in\Z_m\setminus \set{0}}\out{\hat z_i}\bigg)\\
            &=\sum_{z\in \Z_m^n}(1-p)^{|\supp'(z)|}\out{\widehat z}.
        \end{align*}
        More generally, for any $A\subseteq [n]$, and
        $x\in [m]^n$,
        we can define random variable $\noise^A_p(x)$ by applying the noise operator to $x$ on the coordinates in $A$ and leaving other coordinates
        unchanged.
        We then have a corresponding
        $$\fouriernoise^A_p=\identity^{\otimes ([n]\setminus A)}\otimes \bigotimes_{i\in A}\bigg(\out{\hat 0} + (1-p)\sum_{z_i\in\Z_m\setminus \set{0}}\out{\hat z_i}\bigg)\\
            =\sum_{z\in \Z_m^n}(1-p)^{|\supp'(z)\cap A|}\out{\widehat z}.$$

\begin{proposition}
\label{prop:noise-A}
   Let $x, y\in [m]^n$.  If $x_{[n]\setminus A}=y_{[n]\setminus A}$ then,
   \[
                \Pr\s{\noise^A_p(x)=y}
                =\bra y \fouriernoise^{A}_p\ket x
                =
                \p{1-p+\frac{p}{m}}^{|A|-d_A(x,y)}
                \p{\frac{p}{m}}^{d_A(x,y)}\]
                where $d_A(x,y)=\sum_{i\in A} (1-\delta_{x_i,y_i})=d(x,y)$ is the Hamming distance
                between $x_A$ and $y_A$.
                When $x_{[n]\setminus A} \neq y_{[n] \setminus A}$, $\Pr\s{\noise^A_p(x)=y}=0$.
\end{proposition}

The following is the key observation behind our arguments.

\begin{lemma}
\label{lem:fouriernoise}
For $A\subseteq [n]$, $0\le p<1$, and every nonnegative integer $t$,
$\projfouriervs{t}\preceq (1-p)^{-t}\, \fouriernoise^A_p$.
\end{lemma}

\begin{proof}
Both $\fouriernoise^A_p$ and $\projfouriervs{t}$ are diagonal in the Fourier basis.
The diagonal entry at $z\in\Z_m^n$ in
$(1-p)^{-t}\fouriernoise^A_p-\projfouriervs{t}$ is
$(1-p)^{|\supp'(z)\cap A|-t}-1$ if $|\supp'(Z)|\le t$, and
$(1-p)^{|\supp'(z)\cap A|-t}$ otherwise. Thus, every eigenvalue is
nonnegative, and
$\projfouriervs{t}\preceq(1-p)^{-t}\fouriernoise^A_p$.
\end{proof}

We apply this to obtain the following theorem.

\begin{theorem}
\label{thm:noise-events}
            Let $A\subseteq [n]$ and $\event\subseteq[m]^n$ be some event. If $\Pi_\event$ is the projection onto the input states that correspond to event $\event$ then for $0\le t\le n$, we have \[\norm{\Pi_\event\projfouriervs{t}}^2
            \le (1-p)^{-t}
            \max_{x\in \event}\Pr\s{\noise^A_p(x)\in \event}.
            \]
\end{theorem}

        \begin{proof}
        Since $\Pi_\event$ and $\projfouriervs{t}$ are both
        projections, we have
        \begin{align*}
            \norm{\Pi_\event\projfouriervs{t}}^2&=
            \norm{\Pi_\event \projfouriervs{t}\projfouriervs{t}\Pi_\event}\\
            &=
            \norm{\Pi_\event \projfouriervs{t}\Pi_\event}\\
            &\le (1-p)^{-t} \norm{\Pi_\event \fouriernoise^A_p\Pi_\event}\qquad\textrm{by \cref{lem:fouriernoise} and $\Pi_\event, \projfouriervs{t}$ are positive semidefinite}\\
            &\le (1-p)^{-t}\max_{x\in \event}\Pr\s{\noise^A_p(x)\in \event}
            \quad\textrm{by \cref{prop:noise-A} and \cref{thm:perron-frobenius-row-sums}.}\qedhere
        \end{align*}
        \end{proof}

This means that to upper bound the probability of event $\event$ occurring after $t$ queries, it suffices to instead analyze how \emph{stable} $\event$ is under the noise $\noise^A_p$.
We now characterize how this relates to quantum query algorithms.

\paragraph{Noise operators as a tool for bounding quantum query algorithms}

\cref{thm:noise-events} is a very useful tool, although it is a little unclear how exactly it can be applied to quantum query algorithms with full generality.

Let $\domain$ be a finite set and let $\relation\subseteq \domain^n \times \outputspace$ be a relation.
In general we will be interested in quantum algorithms computing functions or relations
on $\domain^n\times \range$, but we sometimes need to analyze the ability of 
quantum query algorithms to compute partial outputs, so we use
$\outputspace$ as a potentially more general notation for some related space which could, of course, be $\range$ itself.

In analyzing quantum algorithms to compute $\relation$, it would be natural to
apply \cref{thm:noise-events} using events
that consist of all inputs
$x\in \domain^n$ that can produce specific
values $y$ such that $(x,y)\in\relation$.
However,
there might be a very small number of inputs $x$ that are much more robust to adding noise than typical choices of $x$ are.
Thus, it is useful to be able to discard a small number of $x$ in a way similar to common methods for proving classical distributional lower bounds, where one gives away success on a very small number of inputs for free. 
That is the intuition for the set $\badset$ in the 
theorem below.


\begin{theorem}[\textbf{Quantum Noise Stability Bound}]\label{cor:noise-stability}
    Let  $\relation \subseteq \domain^n \times \outputspace$.
    Let $A\subseteq [n]$ and $p\in (0,1)$.
    Let $\badset\subset \domain^n$ and, for every $y \in \outputspace$, define $\event_y$ to be the set of $x \in \domain^n\setminus \badset$ such that $(x,y) \in \relation$.
    Then the probability that a quantum query algorithm making $T \leq n$ queries given
    an input chosen uniformly at random from $\domain^n$ 
    produces a $y\in \outputspace$ such that $(x,y)\in \relation$
    is at most
    \begin{displaymath}
        2(1-p)^{-T} \max_{y \in \outputspace} \max_{x \in \event_y} \Pr\s{\noise^A_p(x)\in \event_{y}}  + 2\frac{\abs{\badset}}{\abs{\domain}^n}.
    \end{displaymath}
    When $\badset = \emptyset$ the bound can be tightened to
    \begin{displaymath}
        (1-p)^{-T} \max_{y \in \outputspace} \max_{x \in \event_y} \Pr\s{\noise^A_p(x)\in \event_{y}}.
    \end{displaymath}
\end{theorem}

We note that the minimum over choices of
$A$ in \cref{thm:noise-events} and \cref{cor:noise-stability} always occurs with $A=[n]$, but we state
both more generally since the restricted noise operator will be convenient in applications.
Note that this gives us a way to bound the success probabilities of quantum algorithms with purely classical arguments.

\begin{proof}[Proof of \cref{cor:noise-stability}]
As discussed in the preliminaries when describing quantum queries, we identify
the elements of a set $\domain$ such that $\abs{\domain}=m$ with the set
$[m]$.  
Therefore, in the
following we assume without loss of generality that $\domain=[m]$.

Let $\ket{\psi_T}$
be the final state of a quantum query algorithm after making $T$ queries to a uniformly random input in $[m]^n$.
Let $q$ be the function mapping a measured value in the workspace register of $\ket{\psi_T}$ to the algorithm's produced output.
The probability that the algorithm produces a correct answer for $\relation$ is given by $\norm{\pisucc^\relation \ket{\psi_T}}^2$, where 
\begin{align*}
    \pisucc^\relation 
   &= \sum_{i,p,w} \op{i,p,w}{i,p,w} \otimes \Pi_{\relation^{-1}(q(w))}
\end{align*}
is the projector onto basis states where the algorithm would output a $y$ on input $x$ such that $(x,y) \in \relation$.

We first assume that
$\badset=\emptyset$.
In this case, $\relation^{-1}(q(w))=\event_{q(w)}$.
Since the states  $\ket{i,p,w}$ are mutually orthogonal, the idea is to bound $\norm{\pisucc^\relation \ket{\psi_T}}^2$ by  using \cref{thm:noise-events} applied to the events $\event_{q(w)}$.

To apply this idea, we need one more observation
about $\projfouriervs{T}$.
While the fact that $(\identity_{\queryreg \phasereg \workreg} \otimes \projfouriervs{T})\ket{\psi_T}  = \ket{\psi_T}$ is an immediate consequence of \cref{cor:tquery-fourier}, it is also true that for every choice of $i,p,w$ that $\projfouriervs{T} \ket{\psi^{i,p,w}_T} =\ket{\psi^{i,p,w}_T}$, where $\ket{\psi^{i,p,w}_T}$ is the input state on input register $\inputreg$ resulting from measuring the algorithm registers $\queryreg \phasereg \workreg$ in the standard basis and receiving outcomes $i,p,w$.

Then
\begin{align*}
    \norm{\pisucc^\relation \ket{\psi_T}}^2
    &\leq \max_{i,p,w}\bignorm{\pisucc^\relation \ket{i,p,w} \ket{\psi^{i,p,w}_T}}^2\\
    &=\max_{i,p,w}\bignorm{\Pi_{\relation^{-1}(q(w))}  \ket{\psi^{i,p,w}_T}}^2\\
    &= \max_{i,p,w}\bignorm{\Pi_{\event_{q(w)}} \projfouriervs{T} \ket{\psi^{i,p,w}_T}}^2\qquad\textrm{since we assumed that $\badset=\emptyset$}\\
    &\leq \max_{y\in \outputspace} \bignorm{\Pi_{\event_{y}} \projfouriervs{T}}^2\\
    &\leq (1-p)^{-T}\max_{y \in \outputspace} \max_{x \in \event_y} \Pr\s{\noise^A_p(x)\in \event_{y}}\qquad \textrm{by \cref{thm:noise-events}}.
\end{align*}

We now consider the general case where $\badset\ne \emptyset$.
Then we have $$\Pi_{R^{-1}(q(w))}=\Pi_{\event_{q(w)}}+\Pi_{R^{-1}(q(w))\,\cap\, \badset}$$ and, since 
$\Pi_{R^{-1}(q(w))\,\cap\, \badset}$ is dominated by
$\Pi_{\badset}$ we have,
{\allowdisplaybreaks
\begin{align*}
  &\norm{\pisucc^\relation \ket{\psi_T}}^2\\
    &=\bignorm{\big(\sum_{i,p,w}\op{i,p,w}{i,p,w} \otimes \Pi_{\relation^{-1}(q(w))}\big)\ket{\psi_T}}^2\\
    &= \bignorm{\sum_{i,p,w} \big(\op{i,p,w}{i,p,w} \otimes (\Pi_{\event_{q(w)}}  + \Pi_{R^{-1}(q(w))\,\cap\, \badset})\big)\ket{\psi_T}}^2\\
     &\le \bigg(\bignorm{\big(\sum_{i,p,w}\op{i,p,w}{i,p,w} \otimes \Pi_{\event_{q(w)}}\big)\ket{\psi_T}}+\norm{(\identity_{\queryreg\phasereg\workreg}\otimes \Pi_{\badset})\ket{\psi_T}}\bigg)^2\\
    &\leq \bigg(\max_{i,p,w}\bignorm{\Pi_{\event_{q(w)}}  \ket{\psi^{i,p,w}_T}}+ \norm{(\identity_{\queryreg\phasereg\workreg}\otimes \Pi_{\badset})\ket{\psi_T}}\bigg)^2\\
    &= \bigg(\max_{i,p,w}\bignorm{\Pi_{\event_{q(w)}}  \projfouriervs{T}\ket{\psi^{i,p,w}_T}}+ \sqrt{\frac{\abs{\badset}}{m^n}}\bigg)^2 \quad \textrm{since $\ket{\psi_T}$ has a uniformly random input}\\
     &\le\bigg(\max_{y\in \outputspace} \bignorm{\Pi_{\event_{y}} \projfouriervs{T}}+ \sqrt{\frac{\abs{\badset}}{m^n}}\bigg)^2\\
    &\le 2\max_{y\in \outputspace} \bignorm{\Pi_{\event_{y}} \projfouriervs{T}}^2+2\frac{\abs{\badset}}{m^n}\qquad\textrm{by the Cauchy-Schwarz inequality}\\
    &\leq 2 (1-p)^{-T}\max_{y \in \outputspace} \max_{x \in \event_y} \Pr\s{\noise^A_p(x)\in \event_{y}}  + 2\frac{\abs{\badset}}{m^n} \qquad \textrm{by \cref{thm:noise-events}}.
\end{align*}
}
The statement of the theorem follows since $\abs{\domain}=m$.
\end{proof}
\paragraph{Some simple examples of our quantum noise stability bound}
As simple examples of this method in action, we prove that computing parity on input $x \in \{0,1\}^n$ with probability at least $3/4$ requires at least $n/8$ quantum queries to $x$ and that correctly copying an input $x' \in [m]^n$ with probability larger than $m^{-n/12}$ requires more than $n/3$ quantum queries to $x'$.

Our parity bound is included as a demonstration of how \cref{cor:noise-stability} can be applied in the case of single-output functions, despite the fact that its primary utility will be for functions with many output
values.
Of course, the parity lower bound is already known  (determining the parity of a uniformly random input after fewer than $n/2$ quantum queries is impossible with any probability larger than $1/2$ \cite{PhysRevLett.81.5442, BBC+01}) and the parameters we obtain are far from ideal, but the proof gives an indication of how to apply the method.

\begin{proposition}
    Any quantum algorithm computing $\Parity$ of $n$ uniformly random input bits with $t \leq n/8$ quantum queries to its input has success probability smaller than $3/4$.
\end{proposition}

\begin{proof}
    By \cref{cor:noise-stability}, the success probability of any such quantum algorithm running for at most $t$ steps is bounded above by
    \begin{displaymath}
        (1-p)^{-t} \max_{b \in \{0,1\}}\max_{x \in \Parity^{-1}(b)} \Pr \big[\bigoplus_{i \in [n]}\noise_p(x_i) = b \big].
    \end{displaymath}
    We start by observing that the probability of $\bigoplus_{i \in [n]}\noise_p(x_i) = b$ is exactly the same as the probability that a binomial random variable with $n$ trials and probability $p/2$ is even.
    The probability that such a variable equals to $k$ is $\tbinom{n}{k}(p/2)^k(1-p/2)^{n-k}$.
    Summing this up over all even $k$ gives a total probability of
    \begin{align*}
        \sum_{k \text{ even}} &\binom{n}{k}(p/2)^k(1-p/2)^{n-k}\\
        &= \frac{1}{2}\sum_{k} \binom{n}{k}(p/2)^k(1-p/2)^{n-k} + \frac{1}{2}\sum_{k} \binom{n}{k}(-p/2)^k(1-p/2)^{n-k}\\
        &= 1/2 + (1-p)^n /2 \qquad \textrm{by the binomial theorem.}
    \end{align*}
    Multiplying this by $(1-p)^{-t}$ gives
    $((1-p)^{-t} + (1-p)^{n-t})/2$.  
    Setting $t= n/8$ and $p=1-2^{-2/n}$ this is
    $(2^{1/4}+2^{-14/8})/2<3/4$, proving the bound.
\end{proof}

The following is an alternate proof of a qualitatively similar bound by Boneh and Zhandry (Theorem 5 in \cite{DBLP:conf/eurocrypt/BonehZ13}).

\begin{proposition}\label{lem:identity-query-bound}
    For all integers $n,m \ge 2$, any quantum query algorithm making at most $n/3$ queries to $x$ sampled uniformly from $[m]^n$ succeeds in outputting $x$ with probability less than $m^{-n/15}$. 
\end{proposition}

\begin{proof}
For each $y\in[m]^n$, let $\event_y=\{y\}$. 
Observe that $\noise_p(y)=y$ with probability $(1-p+p/m)^n$.
Applying \cref{cor:noise-stability} on $\cb{E_y}_y$ with $p=1-1/(2m-2)$ and $t=n/3$, we have
\begin{align*}
    \Pr[\text{output=input}]
    &\le (2m-2)^{n/3}\left(1/(2m-2)+ (1-1/(2m-2))/m\right)^n\\
    &=[(2m-2)^{1/3}\ 3/(2m)]^n\\
    &= [3 (1-1/m)^{1/3}/2^{2/3}]^n\cdot m^{-2n/3}.
\end{align*}
Now for $m=2$, direct calculation shows that $3(1-1/m)^{1/3}/2^{2/3}< m^{3/5}$, and this is easier to
satisfy as $m$ increases, so the probability
of producing a correct output is $<m^{3n/5-2n/3}=m^{-n/15}$ as claimed.
\end{proof}


\paragraph{Applications to time-space tradeoff lower bounds}

When proving quantum time-space tradeoff lower bounds using the paradigm of Borodin and Cook~\cite{BC82}, the key property one
needs to show is that after $T$ quantum queries the probability
that any quantum algorithm can produce a partial assignment $\tau$
of at least $Ck$ correct output answers decays exponentially
in $k$.
Therefore, to prove such a lower bound for a function $f$, the standard will be to apply \cref{cor:noise-stability}, with  relation $\relation$ containing all pairs $(x,\tau)$ where $f(x) \| \tau$.

Since our goal is to obtain an exponentially small
upper bound, we choose $p$ such that $(1-p)^{-t}$
is some small exponential in $k$.
In particular, since $1/(1-a+a^2/2)\le e^a$ for every
$a\in [0,1]$,
by choosing $p\le\delta k/t - (\delta k/t)^2/2$, we obtain that $(1-p)^{-t}\le e^{\delta k}$.

Suppose that for this value of $p$, we can choose an $A\subseteq [n]$ and
prove that for every $\tau$ 
\begin{equation}
\label{eq:noise-goal}
\max_{x \textrm{ s.t. } f(x)\|\tau}\Pr\s{f(\noise^A_p(x))\|\tau}\le 2^{-\gamma k},
\end{equation}
where $\gamma>\delta \log_2(e)$.
We then obtain by \cref{cor:noise-stability} that the probability that $\tau$ is 
correct is at most $2^{-(\gamma-\delta \log_2(e))k}$
which satisfies the conditions of the Borodin-Cook key 
lemma for time-space tradeoff lower bounds.

Sometimes we can get tighter bounds by picking a set of problematic inputs $\badset$ containing at most a $2^{-\eta k}$ fraction of all possible inputs and prove that for the above value of $p$ some $A \subseteq [n]$ and every $\tau$
\begin{equation}
    \max_{x \not \in \badset \text{ s.t. } f(x) \| \tau} \Pr\s{f(\noise^A_p(x))\|\tau}\le 2^{-\gamma k},
\end{equation}
again where $\gamma > \delta \log_2(e)$.
By \cref{cor:noise-stability} we then get that the probability that $\tau$ is correct is at most $2^{1-(\gamma - \delta \log_2(e))k} + 2^{1-\eta k}$ which is at most $2^{-\zeta k}$ for some constant $\zeta >0$ when $k$ is larger than some constant.

\section{Universal Hashing}
\label{sec:hashing}

Carter and Wegman~\cite{DBLP:journals/jcss/CarterW79,DBLP:journals/jcss/WegmanC81} defined two closely related
notions of universal hash function families.

\begin{definition}
    A family of functions $\hash$ from $\domain$ to $\range$ is \emph{a universal hash function family} iff for every
    $x\ne x'\in \domain$, $\Pr_{h\in \hash}\s{h(x)=h(x')}=1/|\range|$.
    
    $\hash$ is a \emph{strongly universal} hash function family iff
    for every $x\ne x'\in \domain$ and every $y,y'\in \range$,
    $\Pr_{h\in \hash}\s{h(x)=y \textrm{ and }h(x')=y'}=1/|\range|^2$.   In other words, strongly universal hash
    function families are uniform and pairwise independent.
\end{definition}

Mansour, Nisan, and Tiwari~\cite{DBLP:journals/tcs/MansourNT93} associated a
computational problem, which we denote by $\eval_\hash$, with
each universal hash function family $\hash$ from $\domain=\bits^n$ to $\range=\bits^m$ for $n\ge m$. 
The input to $\eval_\hash$ is a pair $(h,x)$, consisting of an input
$x\in \bits^n$ and a binary description of a hash function $h\in \hash$ such that
a random binary string corresponds to a random
element of $\hash$, and the required output is to compute $h(x)\in \bits^m$.

Mansour, Nisan, and Tiwari~\cite{DBLP:journals/tcs/MansourNT93} proved a general 
time-space tradeoff lower bound of $T=\Omega(nm/S)$ for
any classical algorithm that uses time $T$ and space $S$ to compute
$\eval_\hash$ for any strongly universal\footnote{Though \cite{DBLP:journals/tcs/MansourNT93} simply uses the term ``universal'' from~\cite{DBLP:journals/jcss/CarterW79}, the definition they use requires the pairwise independence of the ``strongly universal'' definition that appears in~\cite{DBLP:journals/jcss/WegmanC81}.} hash function family $\hash$ from $\bits^n$ to
$\bits^m$.

\paragraph{Example applications:}
\begin{itemize}
    \item Mansour, Nisan, and Tiwari~\cite{DBLP:journals/tcs/MansourNT93} proved that the hash function family
$\hash^{n,m}_{\textrm{conv}}=\set{h^{a,b}}_{a,b}$ from $\bits^n$ to $\bits^m$ is strongly universal where $h^{a,b}$ 
for $a\in \bits^{n+m-1}$ and $b\in \bits^m$ is given by
$h^{a,b}(x)=(a\circ x)\oplus b$ where the convolution  $a\circ x=A_a x \bmod 2$ and
$A_a$ is the $m\times n$ Toeplitz matrix defined by
the vector $a$.
This is equivalent to computing the multiplication tableau associated with the integers represented by $a$ and $x$, and outputting the mod-2 sums of each of the middle $m$ columns (the only full columns of the tableau).
(Note that \cite{DBLP:journals/tcs/MansourNT93} use different indexing in their definition.)
\item
Subsequently,
Dietzfelbinger~\cite{DBLP:conf/stacs/Dietzfelbinger96} showed that
the strongly universal property holds for the hash function family $\hash^{n,m}_{\textrm{mult}}=\set{h_{a,b}}_{a,b}$
for $(n+m)$-bit integers $a$ and $b$ where 
$h_{a,b}$  maps $x$ to $ax+b$ and returns the integer consisting of the middle $m$ bits of the result.  That is, the strongly universal property also holds when there are carries.
\end{itemize}

It happens that Abrahamson~\cite{Abr91} had already proven that computing matrix-vector product $f_A(x)=Ax$ over
any non-trivial finite subset of a field has a tight time-space tradeoff lower bound of
$T=\Omega(mn/S)$ for any fixed matrix $A$ that is sufficiently \emph{rigid} (see \cref{sec:matrix-vector} for the definition of rigidity), and that almost all $m\times n$ Toeplitz
matrices are that rigid.  This immediately implies that the same lower
bound applies to computing convolution and hence to computing $\eval_{\hash^{n,m}_{\textrm{conv}}}$; indeed it shows that there are single fixed hash functions
$h^{a,0}\in \hash^{n,m}_{\textrm{conv}}$ that are hard to compute.

Dietzfelbinger later remarked~\cite{DBLP:conf/birthday/Dietzfelbinger18} that the strong universality of $\hash_{\textrm{mult}}$ implies a
time-space lower bound of $T=\Omega(n^2/S)$ for integer multiplication
by use of the theorem of \cite{DBLP:journals/tcs/MansourNT93}, though
he does not give the details.

\begin{proposition}[\cite{DBLP:conf/birthday/Dietzfelbinger18}]
\label{prop:multconv}
    Classical algorithms computing $n$-bit integer multiplication
    using space $S$ requires time $T$ that is $\Omega(n^2/S)$. 
\end{proposition}

\begin{proof}
    We first consider integer multiplication.
    For convenience of notation, assume that the input integers are $6n$-bits long.
    For $2n$ bit integers $a$ and $b$ and an $n$-bit integer $x$
    define $6n$-bit integers $u=2^{4n}a +1$ and $v=  2^{4n}b+x$.
    Then $uv=2^{8n} ab + 2^{4n} (ax+b) + x$.
    Observe that $ax+b$ has at most $3n+1\le 4n$ bits so computing the bits of $uv$ includes computing the bits of $ax+b$;
    in particular, this includes the computation of $h_{a,b}(x)$.
    Thus, evaluating any function $h\in \hash^{n,n}_{\textrm{mult}}$ on its input $x$ is a subfunction of $6n$-bit integer multiplication.  
    This implies that $n$-bit integer multiplication using space $S$
    requires time $T$ that is $\Omega(n^2/S)$ by~\cite{DBLP:journals/tcs/MansourNT93}.
\end{proof}

This lower bound for integer multiplication improves
the classical $T=\Omega(n^2/(S\log n))$ lower bound due to Abrahamson~\cite{Abr91} which is derived
from his $T=\Omega(mn/S)$ lower bound for computing 
convolution.

\paragraph{Extension to quantum query algorithms?}

Beame, Kornerup, and Whitmeyer~\cite{BKW26} used a bucketing method to extend Abrahamson's bound to yield a $T=\Omega(mn/S)$ time-space tradeoff lower bound for fully general quantum
algorithms computing matrix-vector product $f_A(x)=Ax$ for any fixed rigid matrix and hence for convolution.   (In \cref{sec:other-bounds}, we give a much simpler proof of the lower bound of~\cite{BKW26} using the noise operator, which is tight in this case.)
While this implies that fully general quantum algorithms have asymptotically the same $T=\Omega(n^2/(S\log n))$ time-space tradeoff lower bound for integer multiplication that Abrahamson proved for classical computation, the following questions remained open:
\begin{itemize}
    \item 
Can the classical time-space tradeoff lower bounds for computing any strongly universal hash function of~\cite{DBLP:journals/tcs/MansourNT93} be extended to quantum computation? 
\item What about a quantum extension of the tight lower bound for integer multiplication in~\cref{prop:multconv}?
\end{itemize}

The key property of strongly universal hash functions used in the classical lower bound of~\cite{DBLP:journals/tcs/MansourNT93} is that they satisfy a
Hash Mixing Lemma. Informally, this lemma says that if $\hash$ is a strongly universal hash function family
from $\domain$ to $\range$ and one fixes a set
$A\subseteq\hash$, a set $B\subseteq \domain$, and a set $C\subseteq \range$, with $|A\times B|$ larger than $|\range|$, then for uniformly chosen  $h\in A$ and $x\in B$ the probability that $h(x)\in C$ is close to its expectation.

Neither the progress methods of \cite{HM23}, nor the bucketing methods of~\cite{BKW26} can make use of this
hash mixing property: There is no natural notion of unusual progress required for the former and the latter bucketing method requires that very small portions of the input can completely change the output when the rest is fixed, which is incompatible with the hash mixing property since it only applies when large portions of the input are allowed to vary (at least $n$ input bits).

The quantum time-space tradeoff lower bound for convolution-based hashing applies to computing most single members of the family individually.
Quite surprisingly, we are able to use our noise stability bound to prove the exact same fully general quantum time-space tradeoff lower bound for computing most members of arbitrary strongly universal hash function families that match the optimal classical lower bounds.
Also, as a particular consequence, we obtain optimal quantum time-space
tradeoff lower bounds for integer multiplication.

\subsection{Quantum and classical time-space tradeoffs for computing strongly universal hash functions} 

We first state our main lower bound for quantum and classical computation of members of strongly universal hash function families.

\begin{restatable}{theorem}{hashtslb}
\label{thm:almost-all-hash-time-space}
Let $\hash$ be a strongly universal family of hash functions from $\bits^n$ to
$\bits^m$ for $n\ge m$ for sufficiently large $m$.  Let 
$S$ be a space bound and $\ell=\min(S,m/1024)$.   
For all but a $2^{-9\ell}$ fraction of functions $h\in \hash$, any quantum query algorithm that
computes the function $h$ using $S$ qubits of memory, and succeeds with probability at least $2^{-\ell}$ on inputs uniformly chosen from $\bits^n$,
must make a number of queries $T$ that is $\Omega(mn/S)$.
Moreover, this statement also holds if functions identified as hard for space $S$ must be hard for all space bounds $S'$ such that $S \leq S'$.
\end{restatable}

As we have observed, the specialization of this theorem to classical algorithms is new and improves on the results of Mansour, Nisan, and Tiwari~\cite{DBLP:journals/tcs/MansourNT93} in that individual hash functions are shown to be hard, rather than only showing the collective hardness of the 
family as a whole used in their bounds.
The following corollary states the weaker quantum analogue of their classical bound
and follows immediately from~\cref{thm:almost-all-hash-time-space}.

\begin{corollary}
\label{cor:quantum-mnt}
    For any strongly universal family of hash functions $\hash$ from $\bits^n$ to $\bits^m$ for
    $n\ge m$ for $m$ sufficiently large and $\ell=\min(S,m/1024)$, any quantum query algorithm that computes $\eval_\hash$ using 
    $S$ qubits of memory that succeeds with probability at least $2^{-\ell}$ on input $(h,x)$ for
    $h$ uniformly chosen from $\hash$ and $x$ uniformly chosen from $\bits^n$
    must make a number of queries $T$ that is $\Omega(mn/S)$.
\end{corollary}

The following corollary for multiplication is also nearly immediate for the same reason as~\cref{prop:multconv}.

\begin{corollary}
    There is a fixed constant $c$ such that 
    any quantum algorithm that computes $n$-bit integer multiplication
    using space $S \leq cn$ requires time $T$ that is $\Omega(n^2/S)$. 
    Moreover, for at least $(1-1/n^{9})\cdot 2^{2\floor{n/6}}$ choices of $a\in \bits^n$, this bound also holds for computing the function $f_a:\bits^n\rightarrow \bits^{2n}$ given by $f_a(x)=ax$. 
\end{corollary}

\begin{proof}[Proof sketch]
We get this number of hard choices of $a$ by observing that $f_{a,b}(x) = ax +b$ is a strongly universal hash family. Thus at most a $1/n^{9}$ fraction of choices of $a,b$ are excluded from~\cref{thm:almost-all-hash-time-space}.
Let $n'=\floor{n/6}$.
As long as an $a \in \bits^{2n'}$ is not excluded for some choice of $b \in \bits^{2n'}$, it makes the function $f_{2^{4n'} a + 1}$ hard to compute on the input distribution $2^{4n'}b + x$ by using the construction from \cref{prop:multconv}.
Thus at most a $1/n^{9}$ fraction of the choices of $a$ must be excluded.
Since there are $2^{2n'}$ possible choices of $a$, we get at least $(1-1/n^{9}) 2^{2n'} \leq (1-1/n^9) 2^{2\floor{n/6}}$ hard choices of $a$.
\end{proof}

We derive the proof of \cref{thm:almost-all-hash-time-space} using the following key lemma  that lets us apply the
general Borodin-Cook paradigm to prove quantum time-space tradeoff lower bounds for hashing problems.
Its proof uses our quantum noise stability bound together with our
Almost-All Hash Mixing Lemma and is deferred to the next subsection.

\begin{restatable}{lemma}{lemhashlb}
\label{lem:strong-hash-query-bound}
There is an $\alpha>0$ such that the following holds. Let $\hash$ be a strongly universal hash function family from $\bits^n$ to $\bits^m$ for $m\le n$ and $11\le k\le n/32$. 
For all but a $2^{-k/3}$ fraction of $h\in \hash$, 
any quantum query algorithm making
    $t\leq \alpha n$ queries to input 
    $x\in\bits^n$ chosen uniformly at random produces a partial output of size at least $k$ that is consistent with
    $h(x)$ with probability at most $2^{-k/8}$.
\end{restatable}

\begin{proof}[Proof of \cref{thm:almost-all-hash-time-space} using \cref{lem:strong-hash-query-bound}]
Because of the claimed bound we can assume without loss of generality that $T< nm$.  
Further, applying \cref{lem:strong-hash-query-bound} with $k=\floor{m/32}$ we derive that for all but a $2^{-\floor{m/32}/3}\le 2^{-10\ell}$
fraction of hash functions in $\hash$ any quantum query algorithm with
$T\le \alpha n$ and no bound on space has success probability at most
$2^{-\floor{m/32}/8}\le 2^{-\ell}$, so, given the form of the claimed lower bound, we can assume without
loss of generality that $\ell=S\le m/1024$.
Let $\hash^{\emptyset}_{bad}$ be this excluded set of hash functions.
For each $h\in \hash$ let $\calA_h$ be a quantum query algorithm that makes $T$ queries
and uses space $S$.  
Since each algorithm $\calA_h$ requires a register with qubits to index the $n$ inputs values,
we can assume that $2^{S}\ge n$.

We will prove that for all but a $2^{-10S}$ fraction of $h\in \hash$, the probability that $\calA_h$ computes
$h$ on a uniformly random $x$ in $\bits^n$ is at most $2^{-S}$.
Let $\alpha>0$ be the constant from~\cref{lem:strong-hash-query-bound} and let
$t=\floor{\alpha n}$.
We break each $\calA_h$ into $r=\ceil{T/t}$ segments of at most $t$ consecutive queries each.

Choose $k=32S\le m$.  If $T\ge mt/(2k)$ then $T$ is $\Omega(mn/S)$.  
Therefore, we assume that $T< mt/(2k)$ and prove that the algorithms $\calA_h$ must
be incorrect for all but a $2^{-10S}$ fraction of choices of $h\in \hash$.
Likewise it must be the case that $T > t$ as otherwise directly applying \cref{lem:strong-hash-query-bound} to the entire algorithm would give a success probability smaller than $2^{-S}$.

Since each hash function $h$ requires $m$ output bits to be produced, if $\calA_h$ correctly computes $h(x)$ then, one of its $r$ segments must produce partial outputs
of at least $m/r\ge mt/(2T)\ge k$ bits
that are consistent with $h(x)$.
Let $\hash_{bad}$ be the at most $2^{-k/3}\le 2^{-10S}$ fraction of $h\in \hash$ excluded by \cref{lem:strong-hash-query-bound} for our fixed choice of $k$.
By \cref{lem:strong-hash-query-bound}, for all $h \not \in \hash_{bad}$, any quantum algorithm with $t=\floor{\alpha n}$ quantum queries to a uniformly random $x \in \bits^n$ can produce $k$ correct outputs of $h(x)$ with probability at most $2^{-k/8}= 2^{-4S}$.
By \cref{thm:quantum-union-bound} this implies that any fixed segment $j$, which has at most $S$ qubits of input dependent advice from the previous segment, can produce $k$ correct outputs of $h(x)$ when $h \not \in \hash_{bad}$ with probability at most $2^{-3S}$.
Taking a union bound over the $r \leq T < mn \leq 2^{2S}$ segments gives the probability that any segment produces at least $k$
bits that are consistent with $h(x)$ when $h \not \in \hash_{bad}$ is less than $2^{-S}$.
Thus the success probability of $\calA_h$ when $h \not \in \hash_{bad}$ is also less than $2^{-S}$.
Therefore, every element of $\hash$ not in $\hash_{bad}$, which is all but a $2^{-10S}$ fraction
of $\hash$, must have 
$T> mt/(2k)$ and hence $T$ is $\Omega(mn/S)$ as claimed.

As described so far, the choice of the space bound $S$ is fixed and the set of hash functions $\hash_{bad}$ to which the lower bound does not apply might depend on $S$, which is really $\hash^S_{bad}$.   These sets might be  unrelated to each each other. 
However, we can apply the above argument simultaneously for every integer space bound $S\le m/1024$ with $2^S\ge n$ and define $\hash^{\ge S}_{bad}=\hash^{\emptyset}_{bad} \cup \bigcup_{S\le S'\le m/1024}\hash^{S'}_{bad}$. 
We have $|\hash^{\ge S}_{bad}| / |\hash| \le 2^{-10 \ell} + \sum_{S\le S'\le m/1024} 2^{-10S'}< 2^{2-10S} \leq 2^{-9 \ell}$. 
Every hash function not in $h\in \hash^{\ge S}_{bad}$
is hard for every space bound $\geq S$.
\end{proof}

\begin{remark}
    A hash function $h^{a,b}\in \hash_{\textrm{conv}}$ has a near-maximal time requirement for all space bounds if the matrix $A_a$ is highly rigid, which Abrahamson~\cite{Abr91} proves holds with probability exponentially close to 1 as a function of $m$.
    On the other hand, \cref{thm:almost-all-hash-time-space} allows the possibility that a strongly universal family $\hash$ might contain an $n^{-O(1)}$ fraction of functions that do not require time $\Omega(mn/\log n)$ for algorithms with $O(\log n)$ qubits, though the fraction of functions evading maximal time requirements must be much smaller at
    larger space bounds.   This property is inherent in our proof, but may not be necessary.
\end{remark}

\subsection{Mixing for almost all hash functions and the proof of \cref{lem:strong-hash-query-bound}}

We start by formally defining the notion of the smoothness of a probability distribution, which will be useful in our arguments.

\begin{definition}
    For a set $G$, we use $\probspace(G)$ to denote the set of probability distributions on $G$.
    We say that a distribution $\calG\in \probspace(G)$ is \emph{$s$-smooth} iff the collision probability of
    independent samples from $\calG$ is at most $s$.
\end{definition}

While the motivation in~\cite{DBLP:journals/tcs/MansourNT93} was to prove classical time-space tradeoff lower bounds, this lemma does not seem useful for proving lower bounds for quantum computation.   However with our noise operator method, the following closely related lemma (whose proof has much in common with that of the Hash Mixing Lemma) is the key to our lower bounds.

The key differences between our lemma and the hash mixing lemma are as follows:
\begin{itemize}
    \item The mixing property is shown to hold for all but a small number of the functions $h\in \hash$ individually; there is no longer a restricted set $A\subseteq \hash$.
    \item The single target output set $C$ is replaced by a family of such sets.
\end{itemize}
The other less-substantial change, which matters for our application, is that the uniform distribution over a sufficiently large set $B$ is replaced by any distribution on $\domain$ that is sufficiently smooth; i.e., has small collision probability.

\begin{lemma}[Almost-All Hash Mixing Lemma]
\label{lem:avg-hash-mixing}
Let $\hash$ be a strongly universal family of hash
functions from $\domain$ to $\range$.
Let $s\in (0,1)$.
Let $\calB\in \probspace(\domain)$ be a distribution on $\domain$ that is
$s$-smooth.
Suppose that there is a collection $\calC$ of subsets $C\subset\range$ with $|C|/|\range|=q$.
Then, the probability over $h\sim \hash$ that
there exists a $C\in \calC$ such that 
$$\big|\Pr_{x\sim \calB}[h(x)\in C]- q\big|\ge  \sqrt{q(1-q)}$$ is at most 
$s|\calC|$.
\end{lemma}

\begin{proof}
We begin with a second moment calculation associated with the distribution $\calB$ that is very similar to one in the proof of the Hash Mixing Lemma.

\begin{claim}
\label{claim:doublehash}
Fix any $C\in \calC$.  Then
    $\E_{h\sim \hash}\
    \big(\E_{x\sim \calB} (\indicator_{h(x)\in C}-q)\big)^2\le q(1-q)s$.
\end{claim}

\begin{proof}[Proof of Claim]\refstepcounter{equation}
\begin{align*}
    &\E_{h\sim \hash}\
    \big(\E_{x\sim \calB} (\indicator_{h(x)\in C}-q)\big)^2\\
    &=\E_{h\sim \hash}
    \big(\E_{x\sim \calB} (\indicator_{h(x)\in C}-q)\big)\big(\E_{x'\sim \calB}(\indicator_{h(x')\in C}-q)\big)\\
    &=\E_{h\sim\hash} \big(
    \E_{x\sim \calB}\ \E_{x'\sim\calB}\indicator_{h(x)\in C}\, \indicator_{h(x')\in C}\ -\ 2q\E_{x\sim \calB} \indicator_{h(x)\in C}\ +\ q^2 \big)\\
    &=\E_{x\sim \calB}\ \E_{x'\sim\calB}\ \E_{h\sim\hash} \indicator_{h(x)\in C}\, \indicator_{h(x')\in C}\ -\ 2q\E_{x\sim \calB}\  \E_{h\sim\hash} \indicator_{h(x)\in C}\ +\ q^2\\
    &=\E_{x\sim \calB}\ \E_{x'\sim\calB}\ \E_{h\sim\hash}\indicator_{h(x)\in C}\,\indicator_{h(x')\in C}\ -\ q^2\qquad\textrm{since $\E_{h\sim\hash}\indicator_{h(x)\in C}=q$}.\tag{\theequation}\label{eq:mixmoment}
\end{align*}
Now 
\begin{align*}
&\E_{x\sim \calB}\ \E_{x'\sim\calB}\ \E_{h\sim\hash}\indicator_{h(x)\in C}\indicator_{h(x')\in C}\\
&=
\E_{x\sim \calB}\ \E_{x'\sim\calB} \big(\indicator_{x'=x}\cdot \E_{h\sim\hash}\indicator_{h(x)\in C}\  +\ \indicator_{x'\ne x}\cdot\E_{h\sim\hash}\indicator_{h(x)\in C}\ \indicator_{h(x')\in C}\big)\\
&=\E_{x\sim \calB}\ \E_{x'\sim\calB} (q\cdot \indicator_{x'=x}\  +\  q^2\cdot \indicator_{x'\ne x})\qquad\textrm{since $\hash$ is strongly universal}\\
&=q \Pr_{x,x'\sim\calB}[x'=x]\ +\ q^2(1-\Pr_{x,x'\sim \calB}[x'=x]) \\
&\le q s+q^2(1-s) = q(1-q)s+q^2\qquad\textrm{since $\calB$ is $s$-smooth}.
\end{align*}
Plugging this bound into \eqref{eq:mixmoment} yields the claim.
\end{proof}

\begin{claim}
\label{claim:sumhash}
    $\E_{h\sim \hash}\ \sum_{C\in\calC} \big(\E_{x\sim \calB}(\indicator_{h(x)\in C}- q) \big)^2\le q(1-q)s|\calC|$.
\end{claim}

\begin{proof}[Proof of Claim]
\begin{align*}
 \E_{h\sim \hash}
\sum_{C\in \calC} \big(\E_{x\sim \calB}(\indicator_{h(x)\in C}- q) \big)^2
&=\sum_{C\in \calC}\ \E_{h\sim \hash}  \big(\E_{x\sim \calB}(\indicator_{h(x)\in C}- q) \big)^2\\
&\le q(1-q)s|\calC|\qquad\textrm{by \cref{claim:doublehash}}\qedhere
\end{align*}
\end{proof}

\begin{claim}
\label{claim:chebyshev}
    $\Pr_{h\sim\hash} \bigg[\exists C\in \calC\mathrm{\ s.t. } \Pr\big[\big(\E_{x\sim \calB}(\indicator_{h(x)\in C}- q) \big)^2 \ge q(1-q) \big]\bigg]\le s|\calC|.$
\end{claim}

\begin{proof}
    By Markov's inequality, 
    the probability over $h\in \hash$ that $\sum_{C\in\calC} \big(\E_{x\sim \calB}(\indicator_{h(x)\in C}- q) \big)^2\ge q(1-q)$ is at most $s|\calC|$.
    This quantity is the expected number of $C\in \calC$ such that the expected value given $C$ is large, which is an upper bound on the probability that there is such a $C\in\calC$.
\end{proof}

Since
\begin{equation*}
\big|\Pr_{x\sim \calB}[h(x)\in C]- q\big|
=\big|\E_{x\sim \calB}\indicator_{h(x)\in C}-q\big|
\end{equation*}
using \cref{claim:chebyshev},
we have
\begin{equation*}
\Pr_{h\sim\hash} \bigg[\exists C\in \calC\textrm{ s.t. } \big|\Pr_{x\sim \calB}[h(x)\in C]- q\big|\ge \sqrt{q(1-q)} \bigg]\le s|\calC|
\end{equation*}
which is what we wanted to prove.
\end{proof}

We now use the Almost-All Hash Mixing Lemma and our Quantum Noise Stability Bound to prove \cref{lem:strong-hash-query-bound}, which we restate here for convenience.

\lemhashlb*

\blank{
\begin{lemma}
\label{lem:strong-hash-query-bound}
There is an $\alpha>0$ such that the following holds. Let $\hash$ be a strongly universal hash function family from $\bits^n$ to $\bits^m$ for $m\le n$ and $11\le k\le n/32$. 
For all but a $2^{-k/3}$ fraction of $h\in \hash$, 
any quantum query algorithm making
    $t\leq \alpha n$ queries to input 
    $x\in\bits^n$ chosen uniformly at random produces a partial output of size at least $k$ that is consistent with
    $h(x)$ with probability at most $2^{-k/8}$.
\end{lemma}
}

\begin{proof}
The idea is to analyze the problem using the
Noise Stability Bound (\cref{cor:noise-stability}) while
invoking the Almost-All Hash Mixing Bound (\cref{lem:avg-hash-mixing}) for each $x\in \bits^n$ where the distributions $\calB_x$ in 
that lemma are
defined in terms of the noise applied to the input
$x$. 

Let $p\in (0,1)$ to be defined later.
Observe that for $x\in \bits^n$, $\noise_p(x)=x^B$, where $x^B$ is $x$ with the bits of
set $B$ flipped,
and $B$ is randomly chosen
by including each element of $[n]$ independently with
probability $p/2$.
Denote this distribution of subsets of $[n]$ as $\B([n],p/2)$ and
For each $x\in \bits^n$, let $\calB_x$ be the distribution of $x^B$ for $B\sim \B([n],p/2)$ conditioned
on $|B|\ge pn/4$. 
A simple Chernoff bound proves the following:

\begin{claim}
    The probability that $B \sim \B([n],p/2)$ has $\abs{B}<pn/4$ is at most $e^{-pn/16}$; in particular, the probability that $\noise_p(x)$
    is not in the support of $\calB_x$ is at most $e^{-pn/16}$ and every element in 
    the support of $\calB_x$ has probability proportional to its probability under $\noise_p(x)$ by a factor $f>1$.
\end{claim}

\begin{claim}
The distribution $B\sim \B([n],p/2)$ conditioned on $|B|\ge pn/4$ is $2/\binom{n}{pn/4}$-smooth.
\end{claim}

\begin{proof}[Proof of Claim]
Let $B'$ be any subset of $[n]$ with $|B'|\ge pn/4$ be an arbitrary element in the support of the conditional distribution.
Each subset of size $|B'|$ is equally likely, $p \leq 1$, and the probability is monotonically decreasing in the size of the subset, so
 $\Pr_{B\sim \B([n],p/2)}[B=B']< 1/\binom{n}{pn/4}$.   
The conditioning, which holds with high probability, can only increase this probability by a small factor that is certainly at most 2. 
\end{proof}

Since the various $x^B= x^{B'}$ iff $B=B'$, each distribution $\calB_x$ is $2/\binom{n}{pn/4}$-smooth.

Each partial output $\tau$ of $k$ output bits
defines a set $C_\tau\subseteq \bits^m$ of
possible hash function outputs that could be consistent with it.
Since $|\tau|=k$, each $C_\tau$ consists of a fraction
$q=2^{-k}$ of elements of $\bits^m$.
Furthermore, 
the set $\calC$ of possible $C_\tau$
has  $|\calC|=\binom{m}{k}2^k$.

\begin{claim}
\label{claim:apply-hash-mixing}
For each fixed $x\in \bits^n$, at most
a $\binom{m}{k} 2^{k+1} /\binom{n}{pn/4}$ fraction
of $h\in \hash$ have some $C\in \calC$ such that 
$\Pr_{x'\sim \calB_x}[h(x')\in C]\ge 2^{1-k/2}$.
\end{claim}

\begin{proof}[Proof of Claim]
Fix $x$.
Letting $\calB=\calB_x$, $s=2/\binom{n}{pn/4}$, $q=2^{-k}$, we 
apply \cref{lem:avg-hash-mixing}
to derive that the
probability that $h\in \hash$ has some $C\in \calC$
such that 
$$\Pr_{x'\sim \calB_x}[h(x')\in C]\ge  q+\sqrt{q(1-q)}=2^{-k}+\sqrt{2^{-k}(1-2^{-k})}$$
is at most 
$|\calC|s=\binom{m}{k} 2^{k+1}/\binom{n}{pn/4}$.
Since $2^{1-k/2}\ge 2^{-k}+\sqrt{2^{-k}(1-2^{-k})}$ the
claim follows.
\end{proof}

\begin{claim}
\label{claim:noise-hash}
    For each fixed $x\in \bits^n$, at most
a $\binom{m}{k}  2^{k+1}/\binom{n}{pn/4}$ fraction
of $h\in \hash$ have 
$\Pr[h(\noise_p(x))||\tau]\ge e^{-pn/16} + 2^{1-k/2}$
for some $\tau$ that assigns $k$ output
    values.
\end{claim}

\begin{proof}
    This follows immediately from the closeness of
    distributions
    $\noise_p(x)$ and $\calB_x$ and the fact that
    $h(\noise_p(x))||\tau$ iff $h(\noise_p(x))\in C_\tau$ for $C_\tau\in \calC$. 
\end{proof}

Define $\hash_{bad}$ to be the set of those 
$h\in \hash$ such that for at least a $\sqrt{\binom{m}{k}  2^{k+1}/\binom{n}{pn/4}}$ fraction of $x\in \bits^n$,
there is some $\tau$ such that
$\Pr[h(\noise_p(x))||\tau]\ge e^{-pn/16} + 2^{1-k/2}$.
By an averaging (Markov-style) argument, we have
 $|\hash_{bad}|/|\hash| \le \sqrt{\binom{m}{k}  2^{k+1}/\binom{n}{pn/4}}$.

Fix $h\in \hash\setminus \hash_{bad}$.
We will prove that the quantum query algorithm for $h$
running in $t$ steps has small success probability using
\cref{cor:noise-stability}.

Let $\outputspace$ be the set of possible partial
outputs $\tau$ of that assign $k$ bits. 
Let $A=[n]$.
Define relation $\relation^h$ on $\bits^n\times \outputspace$ where $(x,\tau)\in \relation^h$ iff $h(x)$ is consistent with $\tau$, which is the partial output produced by the algorithm for $h$ on input $x$.

Define $\badset^h$ to be the set of $x\in \bits^n$
such that there exists a partial assignment $\tau$ of $k$ bits where $\Pr[h(\noise_p(x))||\tau]\ge e^{-pn/16} + 2^{1-k/2}$.  
For any $\tau\in \outputspace$ define
$\event^h_\tau$ to be the set of $x\in \bits^n\setminus \badset^h$ such
that $h(x)$ is consistent with $\tau$.

We now apply \cref{cor:noise-stability} to say that
the probability that given a uniformly random $x\in \bits^n$, the quantum query algorithm for $h$ making $t\le \alpha n$ queries produces a partial output with at
least $k$ assigned values
that is consistent with $h(x)$ is at most
\begin{align}
&2(1-p)^{-t} \max_{\tau}\max_{x\in \event^h_\tau}\Pr\big[h(\noise_p(x))||\tau\big]+\sqrt{\binom{m}{k} 2^{k+3}/\binom{n}{pn/4}}\nonumber\\
&\le 2(1-p)^{-t} (e^{-pn/16} + 2^{1-k/2})+\sqrt{\binom{m}{k} 2^{k+3}/\binom{n}{pn/4}}.\label{eq:strong-hash-noise-bound}
\end{align}

It remains to optimize the choice of $\alpha$ in the upper bound on $t$, which can be assumed to be maximal $\floor{\alpha n}$ without loss of generality, and to choose a value of $p$ so that \eqref{eq:strong-hash-noise-bound} yields the $2^{-k/8}$ bound claimed in the statement of the lemma.

We choose $p=k/(32t) < 1/2$ which is smaller than $k/(16t) -(k/(16t))^2/2$ so $(1-p)^{-t}\le e^{k/8}$.
Also, since $t= \floor{\alpha n}$, $pn/4\ge k/(128\alpha)$.
Choosing $\alpha=1/256$ we get $2k \leq pn/4 < n/2$ since
$k\le n/32$.

With these values, the bound \eqref{eq:strong-hash-noise-bound}
is 
$$2e^{k/8}(e^{-k/2}+2^{1-k/2})+ \sqrt{\binom{m}{k} 2^{k+3}/\binom{n}{2k}}.$$
Since $m\le n$ and $k\le n/32$ we have 
\begin{align*}
   \binom{m}{k}2^{k+3}/\binom{n}{2k}&\le 2^{k+3}\binom{n}{k}/\binom{n}{2k}
    = 2^{k+3} \frac{2k}{n-k}\cdot\frac{2k-1}{n-k-1}\cdots \frac{k+1}{n-2k+1} \\
    &\le 2^{k+3}\big(\frac{2k}{n-k}\big)^{k}\le 2^{k+3}2^{-3k}= 8\cdot 2^{-2k}.
    \end{align*}
Therefore \eqref{eq:strong-hash-noise-bound}
 is at most
$2e^{k/8}(e^{-k/2}+2^{1-k/2})+\sqrt{8}\cdot 2^{-k}\le 2^{-k/8}$ since $k\geq 11$.

Finally, with these choices of $\alpha$ and $p$ to say that the upper bound on $|\hash_{bad}|/|\hash|$
that is $\sqrt{\binom{m}{k} 2^{k+3}/\binom{n}{pn/4}}\le 2^{1/2-k}<2^{-k/3}$.
\end{proof}

\section{Sorting}\label{sec:sorting}

The time-space tradeoff lower bound for quantum sorting algorithms due
to Klauck, \v{S}palek, and de Wolf~\cite{KSdW07} is the first and most well-known quantum time-space tradeoff lower bound for any problem.
It is a lower bound that is quantitatively of the form $T=\Omega(n^{3/2}/\sqrt{S})$; which matches, up to
a polylogarithmic factor, the existing query algorithms for any space bound $S$ such that $n /\log n \geq S \geq \log^3 n$ \cite{DBLP:conf/stoc/Klauck03}.

While this lower bound has been known for 20 years, it
has a key weakness that is less well known.
The lower bound applies only to quantum algorithms that produce the sorted output in 
sequence, with the time blocks in which the outputs are produced fixed in advance, independent of the input.
Hamoudi and Magniez~\cite{HM23} gave an alternative
proof but both lower bounds only apply to this special case of what we have called
output-oblivious algorithms.

While the restriction on the order of output production (but not the restriction on the time steps in which they are produced) is similar to the optimal classical lower bound of Beame~\cite{DBLP:journals/siamcomp/Beame91} for sorting, it is quite different from the nearly-optimal 
$T=\Omega(n^2/(S\log n))$ general classical lower bound of Borodin and Cook~\cite{BC82} for sorting, which
allows the order in which the output values are
produced to be arbitrarily input-dependent, as is typical
of random-access algorithms; for example, such an algorithm may determine the position in the sorted 
order of certain input values much more quickly than others and be able to output those output positions without having to store those outputs until they can
later be produced.

Beame and Kornerup~\cite{bk:cumulative-journal} extended
the single fixed order lower bound of \cite{KSdW07} to
schemes with arbitrary fixed orders of output production (and thus fixed blocks
in which each output is produced), which is the full case of output-oblivious algorithms. (They also extended the time-space tradeoff lower bound to a
cumulative memory lower bound.)
However, the general technique in the
previous
quantum lower bounds for sorting~\cite{KSdW07,HM23,bk:cumulative-journal}, based on embedding hard functions
in the sorting function depending on which output positions are being produced, is inherently useless in proving lower bounds for general quantum algorithms.

Thus, the problem
of proving any non-trivial fully general quantum lower bound for sorting
has remained open.

While it has seemed appealing to try to extend the specific ideas of the sorting lower bound of Borodin and Cook~\cite{BC82} (not merely their general methods applicable to all functions) to quantum query algorithms, analyzing these ideas
using recording query methods alone seemed unlikely.
However, when we combine this with our method based
on the noise operator and~\cref{cor:noise-stability} for a key part of the analysis
we are able to prove the first time-space tradeoff lower
bounds for sorting using general quantum algorithms. 

\begin{restatable}{theorem}{sortingtradeoff}
\label{thm:sorting-time-space-tradeoff}
    There is a constant $c>0$ such that the following holds.  Let $m\ge n^2$.
    For $\ell=\min(S, cn(\log\log n/\log n)^3)$,
    any quantum query algorithm using $T$ queries and
    space $S$
    that sorts inputs in
    $[m]^n$ with probability at least $2^{-\ell}$ must
    have $T$ that is $\Omega(n^{4/3}\log\log n/(S^{1/3}\log n))$.
\end{restatable}

\paragraph{Proof outline}
Like all time-space lower bounds for multi-output functions based on Borodin and Cook's methods,
the key goal is to show that in any short segment of
the computation that produces the values of roughly $k$ output 
coordinates, it is 
exponentially unlikely that these $k$ output values are
correct for an input drawn from the input distribution.

The specific approach of Borodin and Cook for sorting
is to focus on a property of being
 ``spread out'', every subset of $k\log n$ inputs
contains a subset of size $k$ that has big
gaps between them in the set of possible values $[m]$.
The general idea of their classical lower bound is to show the following:
\begin{enumerate}
    \item The overall success probability in a segment is upper-bounded by the maximum success probability on each fixed computation path.
    \item Except with exponentially small probability
in $k$, the set of variables queried on the path is spread out.
    \item If $\Omega(n)$ input 
    coordinates are not queried on a computation path
    and $2k\log n$ output values are produced,
    either
    \begin{description}
        \item 
    {(i)} most of the values output were never queried, which easily leads to errors almost surely or
    \item{(ii)} most values output were queried, so by the spread-out property there must be
    $k$ output values with big gaps; in this case
    any claim about the output positions - that is, about the exact number of 
    unqueried elements that land in the big gaps between these outputs - is almost surely incorrect.
    \end{description}
\end{enumerate}
An analysis that focuses on computation paths separately is completely incorrect for quantum computation so this proof outline
cannot possibly work.   
Nonetheless, our argument still makes use of a version of the spread-out property.

In place of property 2, for quantum algorithms, using arguments
of Hamoudi and Magniez~\cite{HM23} for the problem of finding many disjoint collisions,
we can show the following property, which holds for a
similar reason to the classical case combined with the properties of oracles in the recording-query basis.
\begin{description}
    \item{A.} In the recording query basis after $t$ steps, the amplitude
    on recording-query basis states where the values
    of the non-$\perp$ components of the state are
    not spread out is exponentially small in $k$.
\end{description}
Although the argument is similar, it means much less
than in the classical case, since the computations on
these recording-query basis states interact
with each other, unlike with classical computation paths.

Since the quantum analysis cannot focus on unqueried input coordinates, we need a different case for the output rather than 3 (i).
Instead, we show the following property whose classical analogue could have been used because it would essentially imply 3 (i) and its failure could have been covered by the proof of 3 (ii).
\begin{description}
\item{B.} Non-spread outputs $\tau$ are exponentially unlikely in any state that is a linear combination of
spread-out recording-query basis states.
\end{description} 
This proof in the quantum case follows the ideas of part of the same recording-query
argument by Hamoudi and Magniez~\cite{HM23} for  disjoint collisions (showing that the quantum algorithm is unlikely to be successful in answering about many collisions that are not already captured with sufficient amplitude in the corresponding recording-query basis states themselves). 

Despite these two parts being natural analogues using recording query arguments, such arguments seem useless
when trying to derive the classical argument for
3 (ii).  
Instead we use bounds involving the noise operator
and \cref{cor:noise-stability} to prove the following
\begin{description}
    \item{C.} Any $t$-query quantum algorithm is exponentially unlikely to be correct in producing
    any particular output $\tau$ that has a 
    subset of size $k$ with big gaps.
\end{description}
This is where the fact that this method is largely classical helps.  The final argument to show this has many aspects in 
common with the classical argument of Borodin and Cook for 3 (ii).
A key difference is that in the Borodin-Cook argument, every unqueried input value has an opportunity to mess up the claimed counts in the big gaps, whereas in our argument, it is only the resampled input values under $\noise_p(x)$
that can mess up these counts.
The fact that only a $p$ fraction of input coordinates is typically resampled
is the main quantitative limitation on our results.

\subsection{Technical overview and analysis of algorithms with few queries}\label{sec:sorting-key-lemmas}

We now proceed to the formal arguments, beginning with the definitions.
We have discussed how a sorting algorithm might 
proceed by choosing to produce outputs in different
positions in the sorted order in an input-dependent
way.   
Any algorithm for sorting that could produce outputs in such an arbitrary order would necessarily have
outputs of the form $(r,v)$ saying that the output in
position $r$ (equivalently, the element of \emph{rank} $r$) has value $v$.
Our formal definition expands these outputs to also include the input
indices $i$ that such values $v$ come from as follows: 

\begin{definition}
    Let $\sortproblem_{n,[m]}$ be the function mapping input $(x_1, \ldots, x_n) \in [m]^n$ to the unique unordered set $\{(i, x_i, r_i)\}_{i \in [n]}$ where $\{r_i\}_{i \in [n]} = [n]$ and $r_i < r_j$ implies that $x_i < x_j$ or both $x_i = x_j$ and $i < j$.
    When the size of the input is known, we drop the
    subscript on $\sortproblem$.
\end{definition}

Note that this definition would be the same as the natural
definition that only produces ranks and values if we expanded each input coordinate to domain $[mn]$ from
domain $[m]$ by appending the input index in the low-order position.
Both definitions also guarantee no ties.
However, our formal definition is nicer because we can
analyze it with a product distribution on $[m]^n$.


For later analysis we find it convenient to define partial outputs as partial functions from input indices
to outputs, rather than as a subset of the triples required for the sorting problem, but the two views are completely equivalent.

       \begin{definition}\label{def:sort-partial-output}
       A \emph{partial output} $\tau=(\tau_v,\tau_r)$ of
       $\sortproblem_{n,[m]}$ is a partial function from $[n]$ to $[m]\times [n]$ that maps each $i\in \dom(\tau)$ to $\tau(i)=(\tau_v(i),\tau_r(i))=(v_i,r_i)$.
       We say that $\tau$ is \emph{consistent with} $\sortproblem_{n,[m]}(x)$, denoted as
       $\sortproblem_{n,[m]}(x)\|\tau$, iff
       $\set{(i,\tau(i))\mid i\in \dom(\tau)}
       \subseteq \sortproblem_{n,[m]}(x)$.
       We write $\abs{\tau}=\abs{\dom(\tau)}$.
    \end{definition}

    Our main target lemma that is in the usual form for the general Borodin-Cook method is the following:

\newcommand{\sortfactor}{8\ceil{\log n/\log\log n}}
\newcommand{\twosortfactor}{16\ceil{\log n/\log\log n}}

\begin{lemma}
    \label{lem:block_query_bound}
    There exist universal constants $\alpha, \gamma > 0$ such that, for all
    sufficiently large $n$ and $m\ge n^2$ and every integer $k$ with $4\le k \le \alpha n/(\sortfactor)^3$, the following
    holds: The probability that any quantum query algorithm making
    $t\leq \alpha k^{2/3}n^{1/3}$ queries to input
    $x\in[m]^n$ chosen uniformly at random produces a partial output of size at least $\sortfactor\ k$ that is consistent with
    $\sortproblem(x)$ is at most $2^{-\gamma k}$.
\end{lemma}

The rest of this section proves this lemma, assuming three key lemmas
corresponding to parts A, B, and C discussed
in our proof outline.   The proofs of those key lemmas
are in \cref{sec:spreadrecord,sec:nonspread,sec:noisysort}.

These proofs depend on whether vectors of values are widely spaced or spread out, which we define now.
   These definitions are similar to those of Borodin and Cook~\cite{BC82}, but we parametrize them more fully and extend them to recording-query basis elements.

    \begin{definition}
        \label{def:spread-out}
        We say that a vector $x_S\in [m]^S$ is
        \emph{$\Delta^{-1}$-widely-spaced} iff 
        $|x_i-x_j|>\frac{m}{\Delta}$ for all $i\neq j\in S$, 
       For $\ell\ge k$ and $I\subseteq[n]$, we say that $x\in [m]^I$ is \emph{$(\ell,\Delta,k)$-spread-out} if, for
        every $A\subseteq I$ of size $\ell$, there exists $S\subseteq A$ of size
        $k$ such that $x_S$ is $\Delta^{-1}$-widely-spaced.
   
        We say that $z\in ([m]\cup \set{\perp})^n$ is
        \emph{$(\ell,\Delta,k)$-spread-out}
        if $z_I\in [m]^I$ is $(\ell,\Delta,k)$-spread-out
        where $I=\set{i\in [n]\mid z_i\ne \perp}$. 

        Finally, we say that partial output $\tau=(\tau_v,\tau_r)$ is $(\ell,\Delta,k)$-spread-out 
        iff the vector given by $(\tau_v(i))_{i\in \dom(\tau)}\in [m]^{\dom(\tau)}$ is $(\ell,\Delta,k)$-spread-out.
    \end{definition}

        \begin{observation}
        \label{obs:nonspread}
        $x\in [m]^I$ fails to be $(\ell,\Delta,k)$-spread-out precisely when there exists some
    $A\subseteq I$ of size $\ell$ such that for every subset $S\subseteq A$ of size $k$, $x_S$
    is not $\Delta^{-1}$-widely-spaced; that some $i,j \in S$ 
    has $|x_i-x_j|\leq\frac{m}{\Delta}$.
    \end{observation}

   \begin{definition}
        Define $\pispread^{\ell,\Delta,k}$ to be the projection that takes input states defined in the recording-query basis and projects them onto the subspace spanned by recording-query basis states
        $\ket{z}$ such that $z$ is
        $(\ell,\Delta,k)$-spread-out.
    \end{definition}

We now state the key lemmas that we will use to prove
the three-part strategy in our outline above.

The first lemma corresponding to part A of our outline is
the following property which we prove in \cref{sec:spreadrecord}:

    \begin{restatable}{lemma}{spreadbasislemma}
        \label{lem:queried-inputs-spread-out}
         Fix some positive integers $k$, $\Delta$, $n$, $m$, $\ell$ and $t$ satisfying $m\ge n^2$, $4\le k\le\Delta\le n$, 
        \begin{math}
            \ell\ge\frac{4k\log\Delta}{\log\log\Delta},
        \end{math}
        and $t\le \sqrt{k\Delta}/12$. 
        Let $\ket{\phi_t}$ be the state of
         a quantum query algorithm (in the recording-query basis) after $t$ queries to a uniformly random input in
        $[m]^n$.  Then
        \begin{math}
            \norm{(\identity-\pispread^{\ell,\Delta,k})\ket{\phi_t}}
            \le \Delta^{-k}.
        \end{math}
    \end{restatable}

The second, corresponding to part B is the following
property which we prove in \cref{sec:nonspread}

    \begin{restatable}{lemma}{nonspreadoutputs}
        \label{lem:non-spread-guesses}
        Suppose that $m\ge 8e$, $n$, $\Delta$ and
        $\ell\geq k$ are integers. 
        Let $\tau$ be a partial output of $\sortproblem_{n,[m]}$
        with
        $\abs{\tau}\ge 2\ell$ whose values are not
        $(2\ell,\Delta,k)$-spread-out and let $\piout^\sortproblem$ be the projection of the state (in the standard basis) onto the subspace 
        where the partial output
        $\tau$ is consistent with $\sortproblem_{n,[m]}$
        on the input.
        Let $\ket{\phi}$ be normalized in the recording-query basis over $[m]^n$.
        Then we have
        $\norm{\piout^\sortproblem \calS \pispread^{\ell,\Delta,k}
                \ket{\phi}} \leq 2^{- \ell/2}$.
    \end{restatable}

    Together, these two lemmas imply the following 
    corollary, which says that any non-spread partial output $\tau$ is exponentially unlikely in $k$ to be correct.
    
    \begin{corollary}
    \label{cor:non-spread-output}
      Fix some positive integers $k$, $\Delta$, $n$, $m$, $\ell$ and $t$ satisfying $m\ge n^2\ge 8e$, $4\le k\le\Delta\le n$, 
        \begin{math}
            \ell\ge\frac{4k\log\Delta}{\log\log\Delta},
        \end{math}
        and $t\le \sqrt{k\Delta}/12$. 
        Then the probability that any algorithm with $t$ quantum queries to a uniformly random input $x\in [m]^n$ produces a partial output $\tau$ 
        with $\abs{\tau}\ge 2\ell$ whose values are not
        $(2\ell,\Delta,k)$-spread-out and that is consistent with $\sortproblem(x)$
         is at most
        $2^{-5k/6}$.
    \end{corollary}

    \begin{proof}
        Let the state of the quantum query algorithm
        after $t$ steps in the standard basis be
        $\ket{\psi_t}$ and let $\ket{\phi_t}=\calS \ket{\psi_t}$.
        Let $\ket{\phi_t^w}$ be the state of $\ket{\phi_t}$ after measuring the work register and reading value $w$.
        We can assume without loss of generality within the context of this lemma that $q(w)$ always gives a partial output $\tau$ that is not $(2\ell,\Delta,k)$-spread-out.
        Let $E$ be the set of such $\tau$.
        Let $\piout^\sortproblem$ be as it was defined in \cref{lem:non-spread-guesses}, and $\pisucc = \sum_{i,p,w} \op{i,p,w}{i,p,w} \otimes \op{q(w)}{q(w)}$.
        Then the square root of the probability that the algorithm produces a correct output that is not $(2\ell,\Delta,k)$-spread-out is given by:
        \begin{align*}
        \norm{\pisucc\ket{\psi_t}}&=
        \norm{\pisucc\calS\ket{\phi_t}}\\
        &=\norm{\pisucc \calS\big((\identity-\pispread^{\ell,\Delta,k})\ket{\phi_t} + \pispread^{\ell,\Delta,k}\ket{\phi_t}\big)}\\
        &\le \norm{\pisucc\calS(\identity-\pispread^{\ell,\Delta,k})\ket{\phi_t}} + \norm{\pisucc\calS\pispread^{\ell,\Delta,k}\ket{\phi_t}}\\
             &\le \norm{(\identity-\pispread^{\ell,\Delta,k})\ket{\phi_t}} + \norm{\pisucc\calS\pispread^{\ell,\Delta,k}\ket{\phi_t}}\\
             &\le \Delta^{-k} + \norm{\pisucc \calS\pispread^{\ell,\Delta,k}\ket{\phi_t}}\qquad\textrm{by \cref{lem:queried-inputs-spread-out}}\\
             &\leq \Delta^{-k} + \max_{w}\norm{\Pi_{succ}\calS\pispread^{\ell,\Delta,k}\ket{\phi_t^w}}\\
             &= \Delta^{-k} + \max_{w}\norm{\Pi_{q(w)}^\sortproblem \calS\pispread^{\ell,\Delta,k}\ket{\phi_t^w}}\\
             &\leq \Delta^{-k} + \max_{\tau, \ket{\phi}}\norm{\piout^\sortproblem \calS\pispread^{\ell,\Delta,k}\ket{\phi}}\qquad \textrm{for $\ket{\phi}$ being a recording query state}\\
             &\leq \Delta^{-k} + 2^{-\ell/2} \qquad \textrm{by \cref{lem:non-spread-guesses}}\\
             &\le 4^{-k} + 2^{-k/2}\le 2^{-5k/12} \qquad\textrm{since $k\le \ell$ and $4\le k\le \Delta$.}
        \end{align*}
        The probability in the statement of the corollary is bounded by the square of this norm.
    \end{proof}

   It remains to handle the remaining case where the
   partial output $\tau$ is spread-out.
This last case is part C of our outline.
The argument is based on the following property when
the noise operator is applied to inputs to $\sortproblem$.  
We prove it in \cref{sec:noisysort}.

 \begin{restatable}{lemma}{noisyspread}
        \label{lem:stable-tie-noise-return}
        Suppose  that $m\ge n^2$ and
        $4\le k\le\Delta\le n/32$, and let $\tau$ be a
         partial output for which there is a
         subset $I\subseteq \dom(\tau)$ of size $k$
         such that $\tau_v(I)$ is $\Delta^{-1}$-widely-spaced.
        For $p\ge 16\Delta/n$ and every $x\in[m]^n$ such
        that $\sortproblem(x)\|\tau$,
        $
        \Pr[\sortproblem(\noise_p(x))\|\tau]
            \le
            e^{-2\Delta}+2^{-(k-2)}$.
    \end{restatable}

We combine this with \cref{cor:noise-stability} to
obtain the following.

    \begin{corollary}
    \label{cor:spreadoutput}
        Suppose  that $m\ge n^2$, and
        $4\le k\le\Delta\le n/32$, and let $\tau$ be a
         partial output with a
         subset $I\subseteq \dom(\tau)$ of size $k$
         such that $\tau_v(I)$ is $\Delta^{-1}$-widely-spaced that is
         produced by a quantum query algorithm that makes at most $t$ queries to a uniformly random
         input $x\in [m]^n$.
        If $16\Delta/n\le \frac{k}{4t}-(\frac{k}{4t})^2/2$, the
        probability that such a $\tau$ is consistent with
        the value of $\sortproblem(x)$ is
        $ \le
            e^{k/4}(e^{-2\Delta}+2^{-(k-2)})$.
    \end{corollary}

\begin{proof}
    If there is no $x\in [m]^n$ such that 
    $\sortproblem(x)\|\tau$ then this probability is 0.
    Otherwise, we apply \cref{cor:noise-stability}
    with $A=[n]$, $p= \frac{k}{4t}-(\frac{k}{4t})^2/2$,
    and $E$ being the event consisting of all those $y\in [m]^n$ such that $\sortproblem(y)\|\tau$, which
    implies
    that this probability is at most
    $(1-p)^{-t}(e^{-2\Delta}+2^{-(k-2)})\le
    e^{k/4}(e^{-2\Delta}+2^{-(k-2)})$.
\end{proof}

With the two cases in hand, we have everything ready to prove our main target  \cref{lem:block_query_bound}. 

\begin{proof}[Proof of \cref{lem:block_query_bound}]
    Let $k$ be an integer such that $4\le k\le t/(\twosortfactor)$ and set $\Delta=\floor{k^{1/3}n^{2/3}}$.
    Then $k\le\Delta\le n/32$.

    \begin{sloppypar}
    Suppose without loss of generality that $t=\floor{ k^{2/3} n^{1/3}/128}$.
    Then $16\Delta/n \le k/(8t)\le \frac{k}{4t}-(\frac{k}{4t})^2/2$ and $t\le \sqrt{k\Delta}/12$.
    Let $\ell=4\ceil{\log\Delta/\log\log\Delta}\ k$.
    Let $\tau$ be a partial output with
    $\abs\tau\ge \sortfactor\ k\ge 2\ell$ produced by
    a quantum query algorithm after $t$ queries
    to a uniformly random input $x\in [m]^n$.
    \end{sloppypar}

    We bound the probability of producing a consistent output by the sum of the probability of producing a consistent output that is $(2\ell, \Delta, k)$-spread-out and the probability of producing a consistent output that is not $(2\ell, \Delta, k)$-spread-out.

If there is a subset $I\subseteq \dom(\tau)$ of size
$k$ such that $\tau_v(I)$ is $\Delta^{-1}$-widely-spaced
then \cref{cor:spreadoutput} implies that
$\tau$ is consistent with $\sortproblem(x)$ with
probability at most $e^{k/4}(e^{-2\Delta}+2^{-(k-2)})$.

If there is no subset $I\subseteq \dom(\tau)$ of size
$k$ such that $\tau_v(I)$ is $\Delta^{-1}$-widely-spaced
then by definition $\tau$ has values that are
not $(2\ell,\Delta,k)$-spread-out and hence
\cref{cor:non-spread-output} implies that the
probability of generating such an output $\tau$ that is also consistent with
$\sortproblem(x)$ is at most $2^{-5k/6}$.

Thus, the probability that the query algorithm given
randomly chosen $x\in [m]^n$ produces a
partial output of size at least $\sortfactor\ k$
that is consistent with $\sortproblem(x)$ is at
most $2^{-5k/6} + e^{k/4}(e^{-2\Delta}+2^{-(k-2)})$
which, since $4\le k\le\Delta$, is at most $2^{-k/12}$. Thus we may take $\gamma=1/12$.
\end{proof}

\subsection{Time-space tradeoff for Sorting}

In this section we use \cref{lem:block_query_bound}
to prove our time-space tradeoff lower bound for sorting, which we restate here for convenience.

\sortingtradeoff*

\begin{proof}
The theorem is trivially true if $T> n^{4/3}$, so we 
assume without loss of generality that $T\le n^{4/3}$.
Since the quantum query algorithm requires an index register $i$ and phase
register $p$ with $\ge \log n$ qubits each, we can
assume that $2^S\ge n^2 > T$.

Let $\alpha$ and $\gamma$ be the constants from the statement of \cref{lem:block_query_bound}.
and define $k=\ceil{3/\gamma} S$.

We break the computation of a quantum query algorithm that makes $T$ quantum queries into segments consisting of $t$ consecutive queries each, where
$t=\lfloor \alpha k^{2/3} n^{1/3}\rfloor$.
There are $\ceil{T/t}$ such segments (where the last segment may be shorter).

Observe that if $T> nt/(\twosortfactor\ k)$ then
$$T\ge \frac{n\lfloor \alpha (3S/\gamma)^{2/3} n^{1/3}\rfloor}{\twosortfactor \ceil{3/\gamma}S}$$ and
hence $T$ is $\Omega(n^{4/3} \log\log n/(S^{1/3} \log n))$.

Likewise, because the input $x$ is part of the output of $\sortproblem(x)$ since it contains triples $(i, X_i, r_i)$ for all $i\in [n]$, \cref{lem:identity-query-bound} implies that
if $T\le n/3$ then the probability of computing 
$\sortproblem$ is at
most $m^{-n/15}<2^{-\ell}$.   
Therefore, we can assume without loss of generality that
$T>n/3$ and
hence, to obtain our claimed lower bound, it suffices to
assume that $S \leq  \alpha n / (\ceil{3/\gamma}(\sortfactor)^3)$ and the choice of $c$ implies that $\ell=S$.

It remains to prove that if we have both $T\le nt/(\twosortfactor\ k)$ and also $S \le \alpha n / (\ceil{3/\gamma}(\sortfactor)^3)$
then the algorithm computes $\sortproblem_{n,[m]}$ with success probability $<2^{-S}$, so we assume 
for contradiction that the success probability of the algorithm
in computing $\sortproblem_{n,[m]}$ is at
least $2^{-S}$.

By picking sufficiently small constant $c$, \cref{lem:block_query_bound} gives us that $T \geq t$ as otherwise the success probability must be smaller than $2^{-\ell}$.
If the algorithm correctly computes $\sortproblem_{n,[m]}$ on an input $x$,
then it must produce $n$ output values in total on
input $x$. 
Therefore, one of the $\ceil{T/t}$ segments must produce partial outputs
of size at least $n/\ceil{T/t}\ge nt/(2T)\ge \sortfactor\ k$
that are consistent with $\sortproblem(x)$.
But, with the appropriate choice of $c>0$, by \cref{lem:block_query_bound}, any given algorithm with $t$ quantum queries can only produce $n/\ceil{T/t}\ge nt/(2T)\ge \sortfactor\ k$ outputs that are consistent with $\sortproblem(x)$ with a probability that is at most $2^{-\gamma k} \leq 2^{-3S}$.
Thus, by \cref{thm:quantum-union-bound}, the probability that any segment can produce such an output with $S$ qubits of advice passed from the previous segment is at most $2^{-2S}$.
Then by taking a union bound over the segments makes the probability that any of them produces $n/\ceil{T/t}\ge nt/(2T)\ge \sortfactor\ k$ correct outputs at most $2^{-2S} \cdot T < 2^{-S} \leq 2^{-\ell}$ as desired.
\end{proof}

\subsection{Almost all amplitude is on spread basis states}
\label{sec:spreadrecord}

In this section we prove the first key lemma, which we restate here for convenience.

\spreadbasislemma*

        \begin{definition}
        Define $\toBin{\Delta}:[m]\to[\Delta]$ by
            $\toBin{\Delta}(v)=j$ iff
            $(j-1)\frac{m}{\Delta}<v\le j\frac{m}{\Delta}$.
        For $S\subseteq[m]$, let
        $\toBins{\Delta}(S)=\{\toBin{\Delta}(v)\mid v\in S\}$. 
        We extend the notation to vectors $z\in ([m]\cup \{\perp\}^n$ by setting $\toBins{\Delta}(z)=\toBins{\Delta}(\set{z_i\mid z_i\ne \perp})$.
    \end{definition}

The following lemma lets us connect spreadness to the analysis of occupied bins.
    
    \begin{lemma}
        \label{lem:cover}
        Fix an integer $\ell\ge k$ and $x_I\in [m]^I$ for $|I|\ge\ell$ such
        that $x_I$ is not
        $(\ell,\Delta,k)$-spread-out. 
        Then there
        exists $A\subseteq I$ with $\abs A=\ell$ such that
        $\abs{\toBins{\Delta}(x_A)}\le 2(k-1)$.
    \end{lemma}

    \begin{proof}
        Let $A$ witness that $x$ is not $(\ell,\Delta,k)$-spread-out, and let
        $B_0=\toBins{\Delta}(x_A)$. 
        Suppose that $\abs{B_0}\ge 2k-1$. 
        For $0<i<k$, let $B_i$
        be the result of removing the two smallest values from $B_{i-1}$. 
        Then
        $\min(B_i)-\min(B_{i-1})>1$ so $\{\min(B_i)\}_{i=0}^{k-1}$ corresponds to $k$ points
        in $x_A$ separated by more than $m/\Delta$. 
        This contradicts that $A$ is a
        witness. 
        So, $\abs{\toBins{\Delta}(x_A)}\le 2(k-1)$.
    \end{proof}

    For $B\subseteq[\Delta]$, and $z\in ([m]\cup\{\perp\})^I$ for $I\subseteq [n]$ define
    \begin{math}
        \counts_B(z)=\big|\{i\in I: \toBin{\Delta}(z_i)\in B\}\big|
    \end{math}.
    Let $\Pi^B_{\ge j}$ and $\Pi^B_{=j}$ project onto recording-query basis states $\ket{z}$ satisfying
    $\counts_B(z)\ge j$ and $\counts_B(z)=j$, respectively.

    \begin{lemma}
        \label{lem:single-query-fixed-progress}
        Let $\ket{\phi}_{\queryreg\phasereg\workreg\functionreg}$ be a state in the recording-query basis.
        Suppose that $\Delta\le m$, and fix $B\subseteq[\Delta]$ and an integer
        $b\ge0$. Then\\
        \centerline{
        $\displaystyle
            \bignorm{\Pi^B_{\ge b+1}\record\ket{\phi}}
            \le
            \bignorm{\Pi^B_{\ge b+1}\ket{\phi}}
            +
            6\sqrt{\frac{\abs B}{\Delta}}
            \bignorm{\Pi^B_{= b}\ket{\phi}}
        $.}
    \end{lemma}

    \begin{proof}
    Since the states $\ket{i,p,w}_{\queryreg\phasereg\workreg}$ for distinct $i,p,w$ are orthogonal and $\record$ does not change the
    values in these registers, it suffices to prove the
    property separately for each fixed value of $i,p,w$.
    Once we fix $i$, \cref{prop:recording-query} implies that
   $\record$ also does not change any register in the recording-query basis other than in the $i$-th input register.
   Therefore it suffices to prove the property when we further fix the value 
   of $z_{[n]\setminus i}\in ([m]\cup\{\perp\})^{[n]\setminus i}$.
   It follows that we can assume without loss of generality that 
   $$\ket{\phi}=\sum_{z_i\in [m]\cup\{\perp\}} \beta_{z_i}
                      \ket{i,p,w}\ket{z}\qquad\textrm{ 
                      with $\sum_{z_i} |\beta_{z_i}|^2=1$.}$$
    For any basis state $\ket{i,p,w}\ket{z}$ in the support of $\ket{\phi}$ as noted above  by \cref{prop:recording-query} any
        $z'$ with
        basis state $\ket{i',p',w'}\ket{z'}$ in the support of
        $\record\ket{i,p,w}\ket{z}$ cannot differ from $\ket{i,p,w}\ket{z}$
        outside of $z'_i$.
        If  $\toBin{\Delta}(z_i)\in B$, $z'_i=\perp$, or $\toBin{\Delta}(z'_i)\notin B$ then $\counts_B(\ket{z'})\le \counts_B(\ket{z})$.
        In the remaining case, in which we have $\counts_B(\ket{z'})=\counts_B(\ket{z})+1$, we must have $\toBin{\Delta}(z_i)\notin B$ and $\toBin{\Delta}(z'_i)\in B$.
        In particular, this implies
        that $\Pi^B_{\ge b+1}\record (\identity-\Pi^B_{\ge b})=0$.
        Therefore 
        \begin{align}
        \norm{\Pi^B_{\ge b+1}\record \ket{\phi}}&\le 
        \norm{\Pi^B_{\ge b+1} \record \Pi^B_{\ge b+1} \ket{\phi}}+ \norm{\Pi^B_{\ge b+1} \record \Pi^B_{=b}\ket{\phi}}\nonumber\\
        &\le \norm{\Pi^B_{\ge b+1} \ket{\phi}}+ \norm{\Pi^B_{\ge b+1} \record \Pi^B_{=b}\ket{\phi}}.\label{eq:counts}
        \end{align}
        Let $S_B=\set{v\in [m]\mid \toBin{\Delta}(v)\in B}$.
        Then $|S_B|< |B|(m/\Delta +1)$.
        From what we have argued,
        $\Pi^B_{\ge b+1}\record\Pi^B_{=b}\ket{i,p,w}\ket{z}=0$
        unless $z_i\notin S_B$ and $\counts_B(\ket{z})=b$.   
Therefore 
$$\Pi^B_{\ge b+1} \record \Pi^B_{=b}\ket{\phi}=\sum_{\substack{
                 z_i \notin S_B\\
        \counts_B(z)=b}}
\beta_{z_i} \Pi^B_{\ge b+1} \record \ket{i,p,w}\ket{z}.$$
Further, for $\ket{i,p,w}\ket{z'}$ to be in the support of $\Pi^B_{\ge b+1}\record\Pi^B_{=b}\ket{\phi}$ we
must have $z'_i\in S_B$.
By \cref{prop:recording-query}, after $\record$ is applied, for each $z'_i\in [m]$ with $z'_i\ne z_i\in [m]$, $\ket{z'_i}$ has amplitude at most $3/m$ in absolute value if $z_i\in [m]$; for $z_i=\perp$ the amplitude of
any $z'_i\in [m]$ is $1/\sqrt{m}$ in
absolute value.
Therefore,
$$\Pi^B_{\ge b+1} \record \Pi^B_{=b}\ket{\phi} 
= \sum_{z'_i\in S_B} \sum_{z_i\notin S_B}\beta_{z_i} \gamma_{z_i,z'_i}\ket{i,p,w}\ket{z'}$$
where $\counts_B(z_{[n]\setminus i})=b$,
$|\gamma_{z_i,z'_i}|\le 3/m$ for $z_i\ne \perp$
and $|\gamma_{\perp,z'_i}|=1/\sqrt{m}$.
It follows that 
{\allowdisplaybreaks
\begin{align*}
    \norm{\Pi^B_{\ge b+1} \record \Pi^B_{=b}\ket{\phi}} 
&= \bignorm{\sum_{z'_i\in S_B} \ \sum_{z_i\notin S_B}\beta_{z_i} \gamma_{z_i,z'_i}\ket{i,p,w}\ket{z'}}\\
&\le \bignorm{\sum_{z'_i\in S_B} \frac{|\beta_{\perp}|}{\sqrt{m}}\ket{i,p,w}\ket{z'}}
+\bignorm{\sum_{z'_i\in S_B}\ \sum_{z_i\in[m]\setminus S_B}\frac{3|\beta_{z_i}|}{m} \ket{i,p,w}\ket{z'}}\\
&= \sqrt{|S_B|}\cdot\frac{|\beta_{\perp}|}{\sqrt{m}}
\ +\ \bignorm{\sum_{z'_i\in S_B}\ \sum_{z_i\in[m]\setminus S_B}\frac{3|\beta_{z_i}|}{m} \ket{i,p,w}\ket{z'}}\\
&\le 3\sqrt{|S_B|}\cdot\frac{|\beta_{\perp}|}{\sqrt{m}}
+3\sqrt{\frac{|S_B|}{m}} \  \big(\sum_{z_i\in[m]\setminus S_B}|\beta_{z_i}|^2\big)^{1/2}\quad\textrm{by Cauchy-Schwarz}\\
&= 3\sqrt{\frac{|S_B|}{m}} \ \bigg(|\beta_{\perp}|+\big(\sum_{z_i\in [m]\setminus S_B}|\beta_{z_i}|^2\big)^{1/2}\bigg)\\
&\le 3\sqrt{\frac{2|S_B|}{m}}\ \norm{\Pi^B_{=b}\ket{\phi}}
\end{align*}
since $\norm{\Pi^B_{=b}\ket{\phi}}= (|\beta_{\perp}|^2+\sum_{z_i\in [m]\setminus S_B}|\beta_{z_i}|^2)^{1/2}$ and
$x+y\le \sqrt{2(x^2+y^2)}$ for
any $x,y\ge 0$.
Combining this with \eqref{eq:counts} and the fact
that $|S_B|/m < |B|(m/\Delta +1)/m\le 2|B|/\Delta$
yields the claim.}
\end{proof}

    \begin{lemma}
        \label{lem:fixed-progress}
        Suppose $\Delta\le m$, and fix $B\subseteq[\Delta]$.
       Let $\ket{\phi_t}$ be the state of a quantum query algorithm after $t$ steps represented in the recording-query basis.
       Then for every
        $0\le b\le t$,
        \\
        \centerline{
            $\displaystyle\norm{\Pi^B_{\ge b}\ket{\phi_t}}
            \le
            \binom{t}{b}
            \left(6\sqrt{\frac{\abs B}{\Delta}}\right)^b$.
        }
    \end{lemma}

    \begin{proof}
    The proof is by induction on $t$ and $b$ using \cref{lem:single-query-fixed-progress}.
        The case $t\ge b=0$ is trivially true.
        For $t\ge b+1\ge 1$, we have $\ket{\phi_t}=\calU_{t}\record\ket{\phi_{t-1}}$.
        Therefore
        \begin{align*}
        \bignorm{\Pi^B_{\ge b+1}\ket{\phi_t}}
        &=\bignorm{\Pi^B_{\ge b+1}\calU_t\record\ket{\phi_{t-1}}}\\
        &=\bignorm{\Pi^B_{\ge b+1}\record\ket{\phi_{t-1}}}\qquad\textrm{since $\Pi^B_{\ge b+1}$ commutes with $\calU_t$}\\
        &\le
        \bignorm{\Pi^B_{\ge b+1}\ket{\phi_{t-1}}}+
        \big(6\sqrt{\frac{\abs B}{\Delta}}\big)\bignorm{\Pi^B_{= b}\ket{\phi_{t-1}}}\qquad\textrm{by \cref{lem:single-query-fixed-progress}}\\
        &\qquad\textrm{where the first term is 0 if $t=b+1$}\\
        &\le \bignorm{\Pi^B_{\ge b+1}\ket{\phi_{t-1}}}+
        \big(6\sqrt{\frac{\abs B}{\Delta}}\big)\bignorm{\Pi^B_{\ge b}\ket{\phi_{t-1}}}\\
        &\le \binom{t-1}{b+1} \big(6\sqrt{\frac{\abs B}{\Delta}}\big)^{b+1}+
        \binom{t-1}{b} \big(6\sqrt{\frac{\abs B}{\Delta}}\big)^{b+1}\quad\textrm{by inductive assumption}\\
        &=\binom{t}{b+1}\big(6\sqrt{\frac{\abs B}{\Delta}}\big)^{b+1},\end{align*}
       which is what we needed to prove.
    \end{proof}

We obtain \cref{lem:queried-inputs-spread-out} by combining this with \cref{lem:cover}.

    \begin{proof}[Proof of \cref{lem:queried-inputs-spread-out}]
        If $\ell>t$,
        the claim follows because the recording database contains at
        most $t$ non-$\bot$ entries after $t$ queries. Suppose $\ell\le t$.
        Since $\ell\ge k\ge 4$,
        $\Delta \ge 144 t^2/k \ge 144 \ell^2/k \ge 144k\ge 64$.
        In particular, $e\le 2(k-1)\le\Delta$ and $\log\Delta\ge 6$.
Then
{\allowdisplaybreaks
\begin{align*}
    &\norm{(\identity-\pispread^{\ell,\Delta,k})\ket{\phi_t}}^2\\
          &\qquad\le
            \sum_{\substack{B\subseteq[\Delta]\\\abs B=2(k-1)}}
            \norm{\Pi^B_{\ge\ell}\ket{\phi_t}}^2\qquad\textrm{by \cref{lem:cover}}\\
            &\qquad\le
            \sum_{\substack{B\subseteq[\Delta]\\\abs B=2(k-1)}}\binom{t}{\ell}^2
            \left(\frac{72(k-1)}{\Delta}\right)^\ell\qquad\textrm{by \cref{lem:fixed-progress}}\\
            &\qquad=
            \binom{\Delta}{2(k-1)}\binom{t}{\ell}^2
            \left(\frac{72(k-1)}{\Delta}\right)^\ell\\
             &\qquad\le
            \bigg(\frac{e\Delta}{2(k-1)}\bigg)^{2(k-1)}
            \left(\frac{ 72(k-1)(et)^2}{\Delta\ell^2}\right)^\ell\\
            &\qquad\le
            \Delta^{2(k-1)}
            \left(\frac{ e^2 k(k-1)}{2\ell^2}\right)^\ell\qquad\textrm{since $2(k-1)\ge e$ and $t\le \sqrt{k\Delta}/12$}\\
            &\qquad\le
            \Delta^{2(k-1)}
            \left(\frac{e^2 (\log\log \Delta)^2}{32(\log \Delta)^2}\right)^{4k\log\Delta/\log\log\Delta}\qquad\textrm{since $\ell\ge 4k\log \Delta/\log\log \Delta$}\\
             &\qquad\le
            \Delta^{2(k-1)}
            \left(\frac{1}{\log \Delta}\right)^{4k\log\Delta/\log\log\Delta}\qquad\textrm{since $\Delta\ge 64$}\\
            &\qquad=\Delta^{2(k-1)}\Delta^{-4k}\le \Delta^{-2k}.\qedhere
\end{align*}}
\end{proof}

\subsection{Non-spread outputs are unlikely to be correct on spread states}
\label{sec:nonspread}
  In this section we prove the second key lemma, which we restate here for convenience:

  \nonspreadoutputs*

    \begin{proof}
        By \cref{obs:nonspread}, there must be some $A\subseteq \dom(\tau)$ of size $2\ell$ such that every subset $S\subseteq A$
        of size $k$ contains distinct $i,j$ with $|\tau_v(i)-\tau_v(j)|\le\frac{m}{\Delta}$.
        We can assume without loss of generality that $\dom(\tau)=A$ and
        hence $\abs{\tau}=2\ell$ since 
        restricting $\tau$ to that witness only enlarges its correctness
        projector.

        We proceed similarly to \cite{HM23}. 
        We define a new family of projectors on recording-query basis states $\Pi^{a,b}_{\tau_v}$ for all $0\le a+b\le 2\ell$ such that $\Pi^{a,b}_{\tau_v}=\sum_z \op{z}{z}$ for
        $z \in \p{[m]\cup\{\perp\}}^n$ such that:
        \begin{enumerate}
            \item \label{item:cond-a} There are $a$ locations $i\in \dom(\tau)$ where $z_i = \perp$.
            \item \label{item:cond-b} There are $b$ locations $i\in \dom(\tau)$ where
                 $z_i\neq\perp$ and $z_i\ne \tau_v(i)$.
        \end{enumerate}

        \begin{claim} 
        \label{claim:abz}
        For $z\in ([m]\cup \{\perp\})^n$,
        \begin{math}
            \norm{\piout \calS \Pi^{a,b}_{\tau_v}
                \ket{z}}
            \leq \big(\frac{1}{\sqrt{m}}\big)^a
            \big(\frac{1}{m}\big)^b
        \end{math}.
        \end{claim}

\begin{proof}[Proof of Claim]
        This follows from the fact that $\mathsf{S}$ maps $\ket{\perp}$ to
        $\sum_{y\in [m]} \frac{1}{\sqrt{m}} \ket{y}$, which yields the $(\frac{1}{\sqrt{m}})^a$ factor, and maps any
        $\ket{y}$ for
        $y \neq \perp$ to $\frac{1}{\sqrt{m}} \ket{\perp} +
            (1-\frac{1}{m}) \ket{y} - \sum_{y'\in [m]\setminus\{y\}} \frac{1}{m}
            \ket{y'}$, which yields the $(\frac{1}{m})^b$ factor.
\end{proof}

    \begin{claim}
    \label{claim:ab}
    For any $\ket{\phi'}$ in the recording query basis,\\
       \centerline{\begin{math}
            \norm{\piout \calS \Pi^{a,b}_{\tau_v}
                \ket{\phi'}}^2
            \leq \binom{2\ell}{a}\binom{2\ell-a}{b}m^{-a-b}
            \norm{\ket{\phi'}}^2.
        \end{math}}
    \end{claim}

    \begin{proof}[Proof of Claim]
        Fix $\Pi_{\tau_v} \succeq \piout$ to be the projector onto inputs where $x_i = \tau_v(i)$ for every $i \in \dom(\tau)$.
        Then $\norm{\piout \calS \Pi^{a,b}_{\tau_v} \ket{\phi'}}^2 \leq \norm{\Pi_{\tau_v} \calS \Pi^{a,b}_{\tau_v} \ket{\phi'}}^2$, 
        so we will now prove the desired inequality for $\norm{\Pi_{\tau_v} \calS \Pi^{a,b}_{\tau_v}\ket{\phi'}}^2$.

        All recording query basis states with distinct values outside $\dom(\tau)$ remain orthogonal after applying $\Pi_{\tau_v} \calS \Pi^{a,b}_{\tau_v}$.
        Let $M = [m] \cup \{\perp\}$.
        If $\ket{\phi'} = \sum_{z \in M} \alpha_z \ket{z}$ then:
        \begin{align*}
            \norm{\Pi_{\tau_v} \calS \Pi^{a,b}_{\tau_v}\ket{\phi'}}^2 & = \sum_{z_{[n] \setminus \dom(\tau)} \in M^{[n] \setminus \dom(\tau)}} \bignorm{\sum_{z_{\dom(\tau)} \in M^{\dom(\tau)}} \alpha_z \Pi_{\tau_v} \calS \Pi^{a,b}_{\tau_v} \ket{z}}^2\\
            &\leq \sum_{z_{[n] \setminus \dom(\tau)} \in M^{[n] \setminus \dom(\tau)}} \big(\sum_{z_{\dom(\tau)} \in M^{\dom(\tau)}} \abs{\alpha_z}^2\big) \sum_{z_{\dom(\tau)} \in M^{\dom(\tau)}} \norm{\Pi_{\tau_v} \calS  \Pi^{a,b}_{\tau_v} \ket{z}}^2\\
            &= \norm{\ket{\phi'}}^2 \sum_{z_{\dom(\tau)}\in M^{\dom(\tau)}} \norm{\Pi_{\tau_v} \calS  \Pi^{a,b}_{\tau_v} \ket{z}}^2.
        \end{align*}
        Observe that the terms in the sum are zero unless $z$ satisfies \ref{item:cond-a} and \ref{item:cond-b} for the choice of $a,b$ and $\tau_v$.
        There are exactly $\binom{2\ell}{a} \binom{2\ell-a}{b}(m-1)^b \leq \binom{2\ell}{a} \binom{2\ell-a}{b}m^b$ such choices for $z$.
        Thus, by combining the above with \cref{claim:abz}, we have:
        \begin{align*}
            \norm{\piout \calS \Pi^{a,b}_{\tau_v}
                \ket{\phi'}}^2 &\leq
            \norm{\ket{\phi'}}^2 \sum_{z_{\dom(\tau)}\in M^{\dom(\tau)}} \norm{\Pi_{\tau_v} \calS  \Pi^{a,b}_{\tau_v} \ket{z}}^2\\
            &\leq \norm{\ket{\phi'}}^2 \cdot  \binom{2\ell}{a} \binom{2\ell-a}{b}m^b \cdot \frac{1}{m^a}
            \frac{1}{m^{2b}}\\
            &= \binom{2\ell}{a} \binom{2\ell-a}{b} m^{-a-b}\norm{\ket{\phi'}}^2.\qedhere
        \end{align*}
    \end{proof}

      \begin{claim}
    \label{claim:abnotsmall}
       For $a+b\le \ell$, 
       $\Pi^{a,b}_{\tau_v}\pispread^{\ell,\Delta,k}\ket{\phi}=0$.
    \end{claim}
    
    \begin{proof}[Proof of Claim]
    Let $\ket{z}$ be a basis state in the support of $\pispread^{\ell,\Delta,k}\ket{\phi}$ that is not mapped to 0 by $\Pi^{a,b}_{\tau_v}$.
    Since only $a+b$ elements
    of $z$ differ from $\tau$, there are at least $2\ell-(a+b)\ge \ell$ coordinates $i\in\dom(\tau)$
    such that $z_i=\tau_v(i)$.  
    Since $\ket{z}$ is $(\ell,\Delta,k)$-spread-out, there is a subset
    $S$ of $k$ of these $\ge \ell$ coordinates such that $|\tau_v(i)-\tau_v(j)|=|z_i-z_j|>m/\Delta$ for all $i\ne j\in S$. 
       This contradicts the non-spreadness property of $\tau$.
    \end{proof} 

We now have the pieces we need to prove of the lemma.
 Since $\Pi^{a,b}_{\tau_v}$ and $\Pi^{a',b'}_{\tau_v}$ are orthogonal for  $(a',b')\not=(a,b)$ and $\sum_{0\le a+b\le 2\ell}\Pi^{a,b}_{\tau_v}=\identity$, we have
{\allowdisplaybreaks
    \begin{align*}
            \bignorm{\piout\calS\pispread^{\ell,\Delta,k}\ket{\phi}} &=\bignorm{\piout\calS
                \big(\sum_{0\le a+b\le 2\ell}\Pi^{a,b}_{\tau_v}\big)
                \pispread^{\ell,\Delta,k}\ket{\phi}}\\
            &=\bignorm{\sum_{\ell+1\le a+b\le 2\ell}\piout\calS
                \Pi^{a,b}_{\tau_v}
                \pispread^{\ell,\Delta,k}\ket{\phi}}\qquad\textrm{by \cref{claim:abnotsmall}}\\
            &\leq \sum_{\ell+1\le a+b\le 2\ell}\bignorm{\piout\calS
                \Pi^{a,b}_{\tau_v}
                \pispread^{\ell,\Delta,k}\ket{\phi}}\qquad\textrm{by the triangle inequality}\\
            &\le \sum_{\ell+1\le a+b\le 2\ell} \sqrt{\binom{2\ell}{a}\binom{2\ell-a}{b}m^{-a-b}}\qquad\textrm{by \cref{claim:ab}}\ \\
            &=\sum_{s=\ell+1}^{2\ell}
            m^{-s/2}\sqrt{\binom{2\ell}{s}}\sum_{a=0}^s\sqrt{\binom{s}{a}}\\
            &\leq \sum_{s=\ell+1}^{2\ell}
            \sqrt{m^{-s}\binom{2\ell}{s} (s+1)2^s} \qquad \textrm{by Cauchy Schwarz}\\
            &\leq \sum_{s=\ell+1}^{2\ell} 2^{\ell+1/2} (e/m)^{s/2} \qquad \textrm{since $(s+1)2^s \leq 2\cdot e^s$ and $\binom{2\ell}{s} \leq 2^{2\ell}$}\\
            &\leq \sum_{s=\ell+1}^{2\ell} 2^{\ell + 1/2} (1/8)^{s/2} \qquad \textrm{since $m \geq 8e$}\\
            &\leq \sum_{s=\ell+1}^{\infty} 2^{\ell + 1/2} (1/8)^{s/2}\\
            &=\frac{2^{\ell + 2}}{(2\sqrt{2})^{\ell + 1}(\sqrt{8}-1)}\\
            &= \frac{\sqrt{2}}{\sqrt{8}-1} 2^{-\ell/2} \leq 2^{-\ell/2}
        \end{align*}
        as desired.
        }
    \end{proof}

\subsection{Noise rarely preserves widely-spaced partial outputs}
\label{sec:noisysort}

\noisyspread*

Our approach to proving this upper bound for a given $\tau$ and $x$ with $\sortproblem(x)\|\tau$ will be to relate it to an anti-concentration property of
Poisson multinomial distributions\footnote{Borodin and Cook used an anti-concentration property of an ordinary multinomial distribution in their argument but that is not quite sufficient here.}.
We give the definition of such distributions with a
notation that will be convenient.
(For some background and many references on Poisson
multinomial distributions see, for example, \cite{DBLP:conf/stoc/DiakonikolasKS16}.)

    \begin{definition}
        Consider $n$ independent trials in which every trial $i\in [n]$ has a
        probability distribution $P_i$ over $[d]$. During trial $i$, event
        $j$ occurs with probability $P_{i,j}$.
        Let $\gamma= \min_{i\in [n], j\in [d]}P_{i,j}$.
        Let
        $C\in\mathbb N^d$ be the count vector, where $c_j$ is the
        number of trials in which event $j$ occurs. The distribution of $C$
        is an $(n,d)$-\emph{Poisson multinomial distribution with minimum probability $\gamma$}.
        \end{definition}

We will use the following anti-concentration bound
on Poisson multinomial distributions.
It follows from an anti-concentration theorem for the sum of $n$ independent Bernoulli random variables~\cite{BCV16}, which is called the Poisson binomial distribution and corresponds to the case
that $d=2$.
We include a proof of a generalization of this bound to $d>2$
in \cref{sec:poisson}.

\begin{restatable}{proposition}{poissonmax}
        \label{lem:poisson-multinomial-point}
        Let $C$ be given by an
        $(n,d)$-Poisson multinomial distribution
        with minimum probability $\gamma$.
        Then every $c\in\mathbb N^d$ satisfies\\
        \centerline{$\displaystyle
            \Pr\s{C=c}
            \le
           \bigg(\frac{e}{2\gamma n}\bigg)^{(d-1)/2}.
        $}
    \end{restatable}

    \begin{proof}[Proof of \cref{lem:stable-tie-noise-return}]
    Let $I\subseteq \dom(\tau)$ be the set of size $k$ such that $\tau_v(I)$ is $\Delta^{-1}$-widely-spaced.
    The general idea of the argument is that 
    the partial output $\tau$ exactly determines the number of elements of $x$ that land in each gap
    of the widely-spaced values in $I$.
    In order for $\tau$  to remain as a
    correct output on input $\noisyx \sim \noise_p(x)$, those numbers must not change.  We show that 
    this is unlikely because of the widely-spaced property.
    
    For $i\in I$,
    write $y_i=\tau_v(i)$ and $r_i=\tau_r(i)$.
    That is, $\tau\| \sortproblem(x)$ iff 
    value $y_i$ has rank $r_i$ in $x$ for each $i\in I$.
    Using the same total ordering as in the sorting problem, we say that $(v_i,i) < (v_j,j)$ exactly when $v_i < v_j$ or both $v_i = v_j$ and $i < j$.
    Since $\tau_v(I)$ is $\Delta^{-1}$-widely-spaced, we can write $I=\{i_1,\ldots,i_k\}$ such that for each $j\in [k-1]$,
    $$y_{i_j}+\ceil{m/\Delta}\le y_{i_{j+1}}$$. 
    In particular, the correctness of the ranks implies
    that for each $j\in [k-1]$,
    there are exactly $r_{i_{j+1}}-r_{i_j}$ values
    of $i\in [n]$ such that $(y_{i_j},i_j)\le (x_i,i)<(y_{i_{j+1}},i_{j+1})$.

    In order for $\tau$ to be consistent with $\noisyx$, for each $j\in [k-2]$ there must also be exactly 
    $r_{i_{j+1}}-r_{i_j}$ values
    of $i\in [n]$ such that $(y_{i_j},i_j)\le (\noisyx_i,i)<(y_{i_{j+1}},i_{j+1})$.
    
    This does not fully characterize consistency; to do so we would also need some other
    conditions:
    \begin{itemize}
        \item For $j\in [k]$, each index $i_j\in I$
        must also have its claimed value -- that is,
        $\noisyx_{i_j}=y_{i_j}=x_{i_j}$
        \item There must be $r_{i_1}-1$ elements
        of $\noisyx$ below $(y_{i_1},i_1)$ and $n-r_{i_k}$ elements above $(y_{i_k},i_k)$
     \end{itemize}
     but we have no guarantee about how large these end gaps are and resampling any element of $I$ is
     relatively low probability, so these conditions would only complicate the analysis without significant benefit.

    For every index $i\in[n]$ and every $j\in[k-2]$ the set of values at position $i$ that
    could be in the $j$-th gap, $G_{i,j}$, by
        \[
            G_{i,j}
            =
            \{v\in[m]:(y_{i_j},i_j)\le(v,i)
            <(y_{i_{j+1}},i_{j+1})\}
        \]
where equality can only occur for $i=i_j$.
For every $j\in[k-2]$; 
observe that
        \[
            \abs{G_{i,j}}
            =y_{i_{j+1}}-y_{i_j}-1
            +\indicator_{i\ge i_j}+\indicator_{i<i_{j+1}}
            \ge (m/\Delta) -1.
        \]
        We define one more set, $G_{i,k-1}$, which contains all of the other possible values for the $i$-th position,
        $
            G_{i,k-1}=[m]\setminus\bigcup_{j=1}^{k-2}G_{i,j}.
        $
        This contains the gap between the last two claimed pairs and so we also have $|G_{i,k-1}|\ge (m/\Delta) -1$.  
        Thus for each $i\in [n]$, the $G_{i,1},\ldots,G_{i,k-1}$ partition $[m]$
        and a uniformly resampled value at index $i$ lies in each
        $G_{i,j}$
        with probability at least $(1/\Delta) - 1/m \geq (1/\Delta) - 1/n \geq 15/(16 \Delta)$.

        Since we must have $\sortproblem(\noisyx)\|\tau$ given
        $\sortproblem(x)\|\tau$, the number of elements
        in each of the gaps for $\noisyx$ must
        be the same as in each of the gaps for $x$.
        Let $R\subseteq[n]$ be the
        random set of resampled indices when sampling
        $\noisyx \sim \noise_p(x)$ (note that this includes the points that did not actually change).
        For each $j\in [k-1]$ define 
        \[
            a_j=\abs{\{i\in R\ :\ x_i\in G_{i,j}\}},
            \qquad
            b_j=\abs{\{i\in R\ :\ \noisyx_i\in G_{i,j}\}}.
        \]
        By the correctness condition, since values outside $R$ are unchanged, we must have
        $a_j=b_j$ for each $j\in [k-2]$, which also
        automatically implies that $a_{k-1}=b_{k-1}$.
        
        For each random choice of the set
        $R$, and each $i\in R$, the distribution over $\noisyx_i$ induces a probability
        distribution $P^R_i$ on $[k-1]$ given by the unique value
        $j\in [k-1]$ such that $\noisyx_i\in G_{i,j}$, each of which occurs with probability
        at least $15/(16\Delta)$.
        Since the distributions $P^R_i$ for distinct
        $i\in R$ are independent, the vector $(b_j)_{j\in [k-1]}$ is
        given by a $(|R|,k-1)$-Poisson multinomial
        distribution with minimum probability $\ge 15/(16\Delta)$.
        Therefore by \cref{lem:poisson-multinomial-point} we have 
        $$\Pr[(b_j)_{j\in [k-1]}=(a_i)_{i\in [k-1]}]\le \bigg(\frac{8e\Delta}{15|R|}\bigg)^{(k-2)/2}.$$
        which upper bounds the
        probability that $\sortproblem(\noise_p(x))\|\tau$ conditioned on $R$ being the re-sampled set.
        
        By definition we have
        $|R|\sim\operatorname{Binomial}(n,p)$ for
        $p\ge 16\Delta/n$.
        The Chernoff bound
        \cref{thm:chernoff-lower-tail} gives
        $\Pr[|R|\le 8\Delta]\le e^{-2\Delta}$.
        Therefore 
        \begin{align*}
     \Pr[\sortproblem(\noise_p(x))\|\tau]&\le 
        \Pr[|R|\le 8\Delta]+\Pr\bigg[\sortproblem(\noise_p(x))\|\tau\ :\  |R|>8\Delta\bigg]\\
        &\le e^{-2\Delta}+(e/15)^{(k-2)/2}\\
        &\le e^{-2\Delta}+2^{-(k-2)}.\qedhere   \end{align*}
    \end{proof}

\section{Other time-space tradeoff lower bounds}\label{sec:other-bounds}

In general, our method using the noise operator is
not so useful if a function depends on its input coordinates in a very biased way in the sense that resampling most coordinates has a high likelihood of having absolutely no effect on the overall function value (let alone on a small set of $k$ output coordinates). 

For example, the uniform distribution over $[m]$ could be used to express the biased distribution of bits
employed
by Hamoudi and Magniez~\cite{HM23} for the $k$-search problem where each coordinate is 1 with probability $q=1/m$ and the goal is to find $k$ 1's.   In our case, we could define the function as finding
$k$ coordinates where the value is $m$.
Note that if the value of a coordinate is already $<m$, which
will be the case for the vast majority of coordinates,
re-sampling that coordinate will only impact the function value in the rare case that the resampling
yields the value $m$.

To prove their output-oblivious lower bound for sorting, \cite{HM23} choose $k=\Theta(S)$ for space bound $S$ and their version of $k$-search requires $q=\Theta(k/n)$, and hence $m=\Theta(n/k)$. 
They choose the value $t$ to be  $\Theta(\sqrt{kn})$.
We would already expect to have only roughly $qn=O(k)$ coordinates
with value $m$.
Re-sampling each such coordinate with probability $p\le k/t=\Theta(\sqrt{k/n})$,
we would expect
to resample only $O(k^{3/2}/n^{1/2})$ of them which will be less than 1 for small space bounds.

Nonetheless, we can reprove some previous lower bounds 
where coordinate values are uniformly useful, in particular we get a much easier proof for matrix-vector product than the existing one in~\cite{BKW26} for the Boolean (or constant sized) domain case, though this new proof cannot take advantage of large
domain size to yield a better lower bound in that case.

\subsection{Matrix-vector products}
\label{sec:matrix-vector}

Here we show how our technique based on the noise
operator yields an alternative proof of the
lower bound on computing matrix-vector products
of Beame, Kornerup, and Whitmeyer~\cite{BKW26}, which
used a bucketing method for recording-query basis states.
The bound is slightly worse since it does not contain the $\Omega(\log d)$
factor where $d$ is the number of possible values in each input coordinate, but the overall argument is much simpler.

We begin with some basic definitions and properties
used in~\cite{BKW26}:

\begin{definition}
An $m \times n$ matrix is \emph{$(g,h,c)$-rigid} iff every $k \times w$ submatrix where $k \leq g$ and $w \geq n-h$ has rank at least $ck$. We call $(g,h,1)$-rigid matrices $(g,h)$-rigid.
\end{definition}

Yesha gives an explicit example of such a matrix and Abrahamson proved that there are many rigid square matrices.

\begin{proposition}[Lemma 3.2 in \cite{Yes84}]\label{prop:DFT-rigid}
The $n\times n$ Discrete Fourier Transform (DFT) matrix is $(n/4,n/4, 1/2)$ rigid.
\end{proposition}

\begin{proposition}[Lemma 4.3 in \cite{Abr91}]\label{prop:rigid-matrices}
There is a constant $\gamma \in (0, \frac{1}{2})$ such that at least a $1-d^{-1}(2/3)^{\gamma n}$ fraction of the matrices over $\inputspace^{n \times n}$ with $|\inputspace|=d$ are $(\gamma n,\gamma n)$-rigid.
\end{proposition}

Our lower bound is based on the following property of
rigid matrices.

\begin{lemma}
\label{lem:rigid-random}
    Let $M$ be an $m\times n$ matrix that is
    $(g,h,c)$-rigid.
    Let $K\subseteq [m]$ where $|K|=k\le g$.
    For $L\subseteq [n]$ with $|L|=k'$ chosen uniformly at random, the probability that the
    submatrix $M_{K,L}$ of $M$ corresponding to the rows of $K$ and the columns of $L$ has rank less than $\min(ck,k'h/(4n))$ is at most $e^{-k'h/(4n)}$.
\end{lemma}

\begin{proof}
Since $M$ is $(g,h,c)$ rigid and $k\le g$, every submatrix of $M_K$ containing at least
$n-h$ columns has rank at least $ck$.
We consider choosing the elements of $L$ iteratively without replacement as $\ell_1,\ldots,\ell_{k'}$
where each $\ell_{i+1}$ is chosen uniformly from $[n]\setminus \set{\ell_1,\ldots,\ell_{i}}$.
For each $i$, we consider the probability that
$\rank(M_{K,\set{\ell_1,\ldots,\ell_{i+1}}})>\rank(M_{K,\set{\ell_1,\ldots,\ell_i}})$.
Let $r=\rank(M_{K,\set{\ell_1,\ldots,\ell_i}})$.

If $r\ge ck$ then we are done.  Otherwise,
$r<ck$.  
Let $C\subseteq [n]\setminus \set{\ell_1,\ldots,\ell_i}$ be the set of all columns $j$ such that
$\rank(M_{K,\set{\ell_1,\ldots,\ell_i,j}})=r+1$.
Now if $|C|\le h$  
the submatrix $M_{K,[n]\setminus C}$ will have $k$ rows and at least $n-h$ columns of $M$ and rank $r<ck$, contradicting the fact that $M$ is $(g,h,c)$-rigid.
This implies that there are at least $h+1$ candidates for
$\ell_{i+1}$ such that $\rank(M_{K,\set{\ell_1,\ldots,\ell_{i+1}}})=r+1$.
Therefore, if the rank is $<ck$ then the probability that the choice of $\ell_{i+1}$ increases the rank by 1
is at least $(h+1)/(n-i)>h/n$.

Since there are $k'$
steps, we can lower bound the rank for a random subset of columns $L$ using a random walk that is non-decreasing on the natural numbers,
starts at 0, moves from $r$ to $r+1$ with probability greater
$h/n$ and is absorbing at $\lceil ck\rceil$.
Ignoring the absorbing barrier, we expect the chain
to reach at least $k'h/n$.
The probability that the Markov chain ends at a value at less than
$\min(ck,k'h/(4n))$ is at most the probability
that $R\sim \operatorname{Binomial}(k',h/n)$ has value 
less than $k'h/(4n)$, which is at most $e^{-k'h/(4n)}$.
\end{proof}

\begin{lemma} \label{lem:matvec}
    Let $M$ be any $(k,h,c)$-rigid $m \times n$ matrix over a field $\F$ and let $f_M:\inputspace^n \to \F^m$ for $\inputspace \subseteq \F$ with
    $|\inputspace|=d$ be defined by $f_M(x) = Mx$. 
    Then for input $x$ sampled
    uniformly from $\inputspace^n$ and any quantum circuit $\calC$ with at most $h/32$ queries to $x$, the probability that $\calC$ produces $k\le h/32$ correct output values of $f_M(x)$ is 
    at most $3\cdot 2^{-ck/4}$.
\end{lemma}

\begin{proof}
    Fix the rigid matrix $M$ and
    let $t=\lfloor h/32\rfloor$.
    Applying \cref{cor:noise-stability} with $A=[n]$ and $\badset = \emptyset$, we obtain that for any $p$,
    the probability any partial assignment $\tau$ to $k$ outputs produced by a quantum query algorithm after $t$ queries agrees with $f_M(x)$ is at most
    $$(1-p)^{-t} \max_{\tau, \abs{\tau} = k} \; \max_{x, f_M(x) \| \tau} \Pr[\ f_M(\noise^A_p(x))\|\tau ].$$

    We choose $p=\delta k/t-(\delta k/t)^2/2\ge \delta k/(2t)$ for some $\delta\in (0,1)$ so that
    $(1-p)^{-t}\le e^{\delta k}$.
    Also, with probability $p$, $\noise_p(x)$ resamples each of 
    the $n$ coordinates of $x$ uniformly from $\inputspace$.
    With this value of $p$, 
    the expected number of resampled coordinates in evaluating 
    $\noise_p(x)$ is 
    $\mu=pn\ge \delta k n/(2t)\ge 16\delta kn/h$.
    By a standard Chernoff bound (e.g., \cref{thm:chernoff-lower-tail}), the probability that the number of resampled coordinates is at most
    $\mu/2=8\delta kn/h$ is at most $e^{-\mu/8}=e^{-2\delta kn/h}\le e^{-2\delta k}$ since $h\le n$.

    Assume that the number of resampled coordinates is at least $k'=\lceil\mu/2\rceil$ and
    let $L$ be a random subset of $k'$ of these resampled coordinates.
    Observe that, conditioned on the number being at least $k'$, the distribution on $L$ is uniform over the subsets of
    $[n]$ of size $k'$.
Therefore, by \cref{lem:rigid-random}, the probability that the rank
of the submatrix $M_{K,L}$ will be less than $\min(ck,k'h/(4n))=\min(ck,2\delta k)$
is at most $e^{-k'h/(4n)}=e^{-2\delta k}$.

Suppose that this does not happen.
Fix the choice of $L$ and fix some $x'$ to the value of $\noise_p(x)$ on $[n]\setminus L$ and value 0 on $L$.
Then $f_M(\noise_p(x))_K=(M x')_K + M_{K,L} z$ where
$z$ is uniformly chosen from $\inputspace^L$.
Then conditioned on $L$ and $x'$, the probability that $f_M(\noise_p(x))_K=\tau$
is at most $1/|\inputspace|^{\rank(M_{K,L})}=d^{-\rank(M_{K,L})}$.
Therefore, the probability 
that $f_M(\noise_p(x))||\tau$ is at most
$d^{-\min(ck,2\delta k)}+e^{-2\delta k}+ e^{-2\delta k}$.

Combining this with our upper bound on $(1-p)^{-t}$, we
have that the total probability that $\tau$ is a
correct output is at most
$$e^{\delta k}(d^{-\min(ck,2\delta k)}+e^{-2 \delta k}+ e^{-2\delta k}).$$
Since $d\ge 2$ setting $\delta=c/2$, this is at most $(e/4)^{c k /2}+2\cdot e^{-ck/2}$ which is
at most $3\cdot 2^{-ck/4}$ since 
$e\le 2\sqrt{2}$.
\end{proof}

We obtain a lower bound that is equivalent to the lower bound of~\cite{BKW26} in the case that $|\inputspace|$ is a constant.   (For simplicity of exposition, we only state it in
the case of square matrices.) 

\begin{theorem}
\label{thm:time-space-matrix-vector}
    There is a constant $C>0$ such that the following holds:
    Let $M$ be an $n\times n$ matrix over a field $\F$ that is $(g,h,c)$-rigid with $n$ sufficiently large.  Then any quantum algorithm using time $T$ and space $S< \min(c g/17, ch/544)$ that computes the function $f_M: \inputspace^n \to \F^n$ for $\inputspace\subseteq \F$ with $|\inputspace|\ge 2$ given by $f_M(x) = Mx$ with success probability larger than $2^{-S}$ requires that $T\ge C\cdot  c n h / S$.
\end{theorem}

\begin{proof}
    Consider any quantum algorithm with $T$ queries and space $S$ that computes $f_M(x)$ with success probability larger than $2^{-S}$.
    For a non-trivial lower bound, we
    only need to consider the case that $T\le n h$.
    We always have $S\ge \log_2 n$ since the query register already takes this space.

  We partition the computation of the quantum algorithm into $\lceil 32T/h\rceil$ segments that each have at most $h/32$ queries.   
  In particular there are at most $32 n+1\le 33n$ such segments.
 
Fix $k = \ceil{16 S / c}$.
Then
$2^{-ck/4} \leq 2^{-4 S} \leq 2^{-2S}/(99 n)$ for
$n$ sufficiently large.
Since $S\le cg/17$, we have $k\le g$ and 
by combining \cref{thm:quantum-union-bound} and \cref{lem:matvec}, we know that each segment of the quantum algorithm  can produce $k\le g$ correct output values with probability at most $3\cdot 2^{S} \cdot 2^{-ck/4}\le 2^{-S}/(33n)$.

Taking a union bound over the at most $33n$ segments, the probability that any of them produces $k$ correct output values is at most $2^{-S}$. Since $f_M$ has $n$ outputs, this means that
    \begin{math}
    \lceil32 T / h\rceil (k-1) \geq n.
    \end{math}
    We know that $T \geq h/32$ and so we have that
    \begin{math}
    64 T / h (k-1) \geq n.
    \end{math}
    Next plugging in our value of $k$ implies
    \begin{math}
    (64 T / h) (16 S / c) \geq n.
    \end{math}
    And so we can conclude that $T\geq  c n h /(1024 S)$
\end{proof}

As noted in~\cite{BKW26}, this theorem has large numbers of corollaries, including that quantum query algorithms computing the discrete Fourier transform with space $S$ require
$T$ queries where $T$ is $\Omega(n^2/S)$ as well as the previous lower bounds for wrapped convolution and integer multiplication~\cite{BKW26} that we discussed
in~\cref{sec:hashing}.

\section{Acknowledgments}
Paul Beame's contribution to this work was supported by National Science Foundation grant CCF-2422205.
Blake Holman's and Niels Kornerup's contributions to this work were supported by the Laboratory Directed Research and Development program (project
240650) at Sandia National Laboratories, a multimission laboratory managed and operated by National
Technology and Engineering Solutions of Sandia LLC, a wholly owned subsidiary of Honeywell International
Inc.~for the U.S. Department of Energy’s National Nuclear Security Administration under contract DE-
NA0003525.
This paper describes objective technical results and analysis. Any subjective views or opinions that might be expressed in the paper do not necessarily represent the views of the U.S. Department of Energy or the United States Government.
The publisher acknowledges that the U.S. Government retains a non-exclusive, paid-up, irrevocable, world-wide license to publish or reproduce the published form of this written work or allow others to do so, for U.S. Government purposes. The DOE will provide public access to results of federally sponsored research in accordance with the \href{https://www.energy.gov/doe-public-access-plan}{DOE Public Access Plan}.

We would additionally like to thank Yassine Hamoudi for insightful discussions regarding the viability of proving fully general quantum time-space tradeoff lower bounds for sorting via the recording query method.

\subsection{Discussion of LLM use}

We began our work by focusing on the general bound for sorting.
After having proved
\cref{cor:non-spread-output}, and independent of our proof of it, we asked ChatGPT 5.6 to solve a 
specific toy example related to 
\cref{sec:noisysort},
involving 
projection onto states of low Fourier degree $\projfouriervs{t}$, that we knew was associated with
the hypergeometric distribution.  
In the course of solving this example, which it did via a convoluted chain
of reasoning involving families of orthogonal polynomials (e.g. \cite{IlievX22}), it connected conjugation by 
$\projfouriervs{t}$
to a noise operator bound for the problem. 
Along the way it proved $(1-p)^{-t}T_p\succeq \projfouriervs{t}$ via an argument involving elementary symmetric polynomials.
We observed that this fact immediately follows from viewing $\projfouriervs{t}$ and $T_p$ in the Fourier basis
and realized that this was likely to be a generally valuable approach that we should apply to other problems.

In retrospect, noise operators are a very natural tool, 
as the properties $T_p$ are straightforward to analyze in both the Fourier basis (where the algorithm's state is easily described) as well as in the standard basis (where our events of interest are defined). While this high-level idea was used in \cref{sec:sorting}, it was only a sub-case of the problem, which still required substantial new analysis to solve in sufficient generality.

Claude Fable 5.1 and GPT 6.1 were used to catch minor errors in drafts of this paper.
The identified errors were verified and fixed by the authors.

All text and proofs contained in this document were written by the authors, who have verified the originality and correctness of the contained results.

\bibliography{sources}
\bibliographystyle{alpha}

\appendix

\section{Poisson multinomial anti-concentration}
\label{sec:poisson}

In this section we prove the following anti-concentration bound
for Poisson multinomial distributions.

    \poissonmax*

\begin{sloppypar}
The proof will be based on the following result of Baillon, Cominetti and Vaisman~\cite{BCV16} which proves strong anti-concentration bounds for sums of independent Bernoulli random variables based on their variances, known as the Poisson binomial distribution, which is the case of  the Poisson multinomial distribution for domain size 2.
\end{sloppypar}

\begin{proposition}[{\cite[Theorem~1]{BCV16}}]
        \label{lem:poisson-binomial-two-outcome}
        Let $P_1,\ldots,P_n$ be $n$ Bernoulli distributions with associated
        probabilities $p_1,\ldots, p_n$.
        If we independently choose
        $b_1\sim P_1,\ldots, b_n\sim P_n$ and write
        $\sigma^2=\sum_{i\in [n]} p_i(1-p_i)$ for
        the variance of their sum,
        then
        every $c \in\mathbb N$ satisfies
        \[
            \sigma\cdot \Pr\big[\sum_{i\in [n]}b_i =c\big]
            \le \eta
        \]
        where $\eta\approx 0.4688\le 1/2$
    \end{proposition}

If we have  $\gamma\le p_i\le 1-\gamma$ then $p_i(1-p_i)\ge \gamma/2$ and hence
$\sigma^2=\sum_{i\in [n]} p_i(1-p_i)\ge \gamma n/2$, which
gives a slightly sharper upper bound of $1/\sqrt{2\gamma n}$ for the
probability in \cref{lem:poisson-multinomial-point}
in the case that $d=2$.

    \begin{proof}
        Let $d\ge 2$.
        Fix $c\in\mathbb N^d$ with $\sum_{j=1}^d c_j=n$, and relabel
        the events without loss of generality so that
        $c_1\le\cdots\le c_d$.
        It will suffice
        to derive good upper bounds on
        $\Pr[C_j=c_j\mid C_{1}=c_1,\ldots, C_{j-1}=c_{j-1} ]$ under this re-ordering
        for each $j\in [d-1]$.
        Observe that, as noted above, the case $j=1$
        follows immediately from~\cref{lem:poisson-binomial-two-outcome} with the slightly 
        better bound of $1/\sqrt{2\gamma n}$, since the
        probability
        is the same as that of the induced Bernoulli distribution in which all values in $[2,d]$
        are indistinguishable.

        For $j>1$, let $s_{j-1}=\sum_{j'\in [j-1]} c_{j'}$ and $n_j=n-s_{j-1}$.
        The conditioning that $C_1=c_1,\ldots, C_{j-1}=c_{j-1}$ means that there is some
        subset $A\in \binom{[n]}{s_{j-1}}$ where 
        $v_i\sim P_i$ has value in $[1,\ldots,j-1]$ for
        every $i\in A$
        and for every $i\in [n]\setminus A$
        the value $v_i\sim P_i$ has a value in
        $[j,\ldots, d]$.
        For each fixed choice of $A$ consistent
        with our conditioning, the conditional probability
        that $C_j=c_j\;|\;C_1=c_1,\ldots C_{j-1}=c_{j-1}$ is
        equivalent to $c_j$ being the sum of $n_j$ independent Bernoulli
        distributions for $i\in [n]\setminus A$ obtained by sampling $P_i \;| \; (P_i\ge j)$, with success on value $j$.
        Observe that such a sample equals $j$
        with probability exactly
        \[
            \frac{P_{i,j}}{\sum_{j'=j}^d P_{i,j'}}.
        \]
        Both this probability and its complement must be at least $\gamma$ since $P_{i,j+1}\ge\gamma$.
 Therefore conditioned
        on $C_1=c_1,\ldots,C_{j-1}=c_{j-1}$ and the
        value of $A$,  
        \cref{lem:poisson-binomial-two-outcome}
        implies that the probability that
        $C_j=c_j$ is at most $1/\sqrt{2\gamma n_j}$.
        Therefore,
        $$\Pr[C_j=c_j\mid C_{1}=c_1,\ldots, C_{j-1}=c_{j-1} ]\le 1/\sqrt{2\gamma n_j}.$$
Since the $c_j$ are non-decreasing, we have
$n_j\ge n-(j-1)n/d=(1-(j-1)/d)n$.
Therefore,
{\allowdisplaybreaks
        \begin{align*}
            \Pr[C=c]
            &=\prod_{j=1}^{d-1} \Pr[C_j=c_j\mid C_1=c_1,\ldots,C_{j-1}=c_{j-1}]\\
            &\le 
            \prod_{j=1}^{d-1} \frac{1}{\sqrt{2\gamma n_j}}\\
            &\le \bigg(\frac{1}{\sqrt{2\gamma n}}\bigg)^{d-1}\  \prod_{j=1}^{d-1} \frac{1}{\sqrt{1-(j-1)/d}}\\
            &=\bigg(\frac{1}{\sqrt{2\gamma n}}\bigg)^{d-1}\ \ \prod_{j=1}^{d-1} \sqrt{\frac{d}{d-(j-1)}}\\
            &=\bigg(\frac{1}{\sqrt{2\gamma n}}\bigg)^{d-1}
            \sqrt{\frac{2d^{d-1}}{d!}}\\
            &\le\bigg(\frac{e}{2\gamma n}\bigg)^{(d-1)/2}
        \end{align*}}
        since $d!\ge 2(d/e)^{d-1}$ for $d\ge 2$.
    \end{proof}

\end{document}